\documentclass[11pt]{article}
\usepackage{geometry}
\usepackage{graphicx}
\usepackage{amsmath}
\usepackage{amssymb}
\usepackage[all]{xy}
\usepackage{xstring}
\usepackage{xspace}

\newtheorem{theorem}{Theorem}[section]
\newtheorem{corollary}[theorem]{Corollary}
\newtheorem{lemma}[theorem]{Lemma}
\newtheorem{proposition}[theorem]{Proposition}
\newtheorem{claim}[theorem]{Claim}
\newtheorem{definition}[theorem]{Definition}

\newtheorem{observation}[theorem]{Observation}

\def\squarebox#1{\hbox to #1{\hfill\vbox to #1{\vfill}}}
\newcommand{\qed}{\hspace*{\fill}\vbox{\hrule\hbox{\vrule\squarebox{.667em}\vrule}\hrule}\smallskip}
\newenvironment{proof}{\noindent{\bf Proof:~~}}{\(\qed\)}

\newcommand{\can}{q} % candidate's strength
\newcommand{\ncan}{r} % candidate's "normalized strength \ncan = \min {1/2,\can}
\newcommand{\dep}{x} % department's strength
\newcommand{\eps}{\epsilon}

 \def\stratpp{{\langle +,+ \rangle}}
 \def\stratpm{{\langle +,- \rangle}}
 \def\stratmp{{\langle -,+ \rangle}}

\newcommand{\xhdr}[1]{\paragraph{\textbf #1}}

\newcommand{\omt}[1]{}
\newcommand{\Xomit}[1]{}
\newcommand{\supproof}[1]{}
\newcommand{\supproofeq}[1]{}

\newcommand{\p}[3]{ % #1 player index , %2 player 1 strength, #3 player 2 strength
    \IfEqCase{#1}{%
        {1}{\frac{#2}{#2+#3}} %
        {2}{\frac{#3}{#2+#3}} %
        }
        }

\newcommand{\btc}{}
\newcommand{\atc}{}

\newcommand{\onec}[1]{#1~}
\newcommand{\twoc}[1]{}

\begin{document}
\title{Dynamic Models of Reputation and Competition in Job-Market 
  Matching}
\author{  
Jon Kleinberg
\thanks{
Cornell University, Ithaca NY.
Email: kleinber@cs.cornell.edu.
%Supported in part by
%a Simons Investigator Award,
%a Google Research Grant,
%an ARO MURI grant, 
%and NSF grants
%IIS-0910664, %HCC
%CCF-0910940, %AF 
%and 
%IIS-1016099. %Jure
}
 \and 
Sigal Oren
\thanks{
Microsoft Research and Hebrew University, Israel.
Email: sigalo@cs.huji.ac.il.
%Supported in part by NSF grant CCF-0910940 and a Microsoft Research Fellowship.
}
}
\date{}

\begin{titlepage}
\maketitle

\begin{abstract}
%!TEX root =recruiting1c.tex
A fundamental decision faced by a firm hiring employees --- and a
familiar one to anyone who has dealt with the academic job market, for
example --- is deciding what caliber of candidates to pursue. Should
the firm try to increase its reputation by making offers to
higher-quality candidates, despite the risk that the candidates might
reject the offers and leave the firm empty-handed? Or is it better to
play it safe and go for weaker candidates who are more likely to
accept the offer? The question acquires an added level of complexity
once we take into account the effect one hiring cycle has on the next:
hiring better employees in the current cycle increases the firm's
reputation, which in turn increases its attractiveness for
higher-quality candidates in the next hiring cycle. These
considerations introduce an interesting temporal dynamic aspect to the
rich line of research on matching models for job markets, in which
long-range planning and evolving reputational effects enter into the
strategic decisions made by competing firms.

The full set of ingredients in such recruiting decisions is complex,
and this has made it difficult
to model the fundamental strategic tension
at the core of the problem.  Here we develop a model based on two
competing firms to try capturing as cleanly as possible the
elements that we believe constitute this strategic tension:
the trade-off between short-term recruiting success and long-range 
reputation-building; the inefficiency that results from 
underemployment of people who are not ranked highest; and
the influence of earlier accidental outcomes on long-term reputations.

Our model exhibits all these phenomena in a stylized setting,
governed by a parameter $q$ that captures the difference in strength
between the top candidate in each hiring cycle and the next best.
Building on an economic model of competition between parties of
unequal strength, we show that when $q$ is relatively low, the
efficiency of the job market is improved by long-range reputational
effects, but when $q$ is relatively high, taking future reputations
into account can sometimes reduce the efficiency. While this trade-off
arises naturally in the model, the multi-period nature of the
strategic reasoning it induces adds new sources of complexity, and
our analysis reveals interesting connections between competition with
evolving reputations and the dynamics of urn processes.

\end{abstract}

\thispagestyle{empty}
\end{titlepage}

%Jon Kleinberg \titlenote{ Supported in part by 
%the John D. and Catherine T. MacArthur Foundation,
%a Google Research Grant,
%a Yahoo!~Research Alliance Grant,
%and NSF grants 
%IIS-0705774, % wya
%IIS-0910664, %HCC
%and CCF-0910940. %AF
%}\\
%       \affaddr{Department of Computer Science}\\
%       \affaddr{Cornell University, Ithaca NY 14853}\\
%       \email{kleinber@cs.cornell.edu}
%% 2nd. author
%\alignauthor
%Sigal Oren \titlenote{Supported in part by 
%NSF grant CCF-0910940 and a Microsoft Research Fellowship.}\\
%       \affaddr{Department of Computer Science}\\
%       \affaddr{Cornell University, Ithaca NY 14853}\\
%       \email{sigal@cs.cornell.edu}
%       }

%!TEX root =recruiting1c.tex
\section{Introduction}
\label{sec:intro}

Markets for employment have been the subject of several large bodies
of research, including the long and celebrated line of work on 
bipartite matching of employers to job applicants \cite{roth-matching},
sociological and economic approaches to the process of finding a job
\cite{granovetter-getting-a-job,
mortensen-search-equilibrium,rogerson-job-search},
and many other frameworks.
Recent work in theoretical computer science has modeled issues
such as the competition among employers for applicants
\cite{immorlica-dueling-algs,immorlica-competing-secretary}
and hiring policies that
take a firm's reputation into account \cite{broder-lake-wobegon}.

Despite this history of research, there remain a number of fundamental
issues in job-market matching that have gone largely unmodeled.
One of these, familiar to anyone who has dealt with job markets
in academia or related professions, is the feedback loop over
multiple hiring cycles between
the job candidates that a firm (or academic department) pursues 
and the evolution of its overall reputation.

There is of course a very broad set of ingredients that go into 
the competition for job candidates over multiple hiring cycles.
This makes it challenging to abstract the basic issues
into a model for this type of multi-period competition.
In the present paper we pare down this complexity to try
formulating a model that captures, as cleanly as possible,
what we view to be the basic sources of strategic tension in the process.

We develop a model based on two firms
that compete for candidates over multiple periods, with a pool of candidates
that has the same structure in each period; the outcome of
the competition for a given candidate is determined probabilistically,
based on the relative reputations of the two firms at the time they compete.
This is a highly reduced and stylized model, but it produces
a complex set of behaviors that we believe should be components
of any of a range of richer extensions of the model as well.
In particular, the process of job-market competition in our model
exhibits the following fundamental trade-offs:
\begin{itemize}
\item[(i)] Successfully recruiting higher-quality candidates
can raise a firm's reputation, which in turn can make it more attractive
to candidates in future hiring cycles.
\item[(ii)] On the other hand,
competing for these higher-quality candidates comes with a greater risk of
emerging from a given hiring cycle empty-handed.
\item[(iii)] The incentive to compete for top-ranked candidates 
can lead to underemployment of lower-ranked candidates,
as they are at risk of receiving no offers while firms instead compete
for their higher-ranked counterparts.
\item[(iv)] The trajectory of the process can be heavily influenced
by a small number of ``accidental'' recruiting outcomes in the early stages,
as reputations are first being established.
\end{itemize}

The trade-off between (i) and (ii) above arises
from the equilibrium in the dynamic matching game played by the two firms
with respect to the pool of available candidates.
We find that in equilibrium, there is an initial period of competition,
which can end when one firm decides it is so far behind the other that
it is no longer worth competing for the top-ranked candidates.
For certain natural ranges of parameters, there is in fact an 
interesting bifurcation --- depending on the random outcomes of the
initial stages, there is a positive probability one firm will ``give up''
on competing for the best candidates, but also a positive probability
that the two firms will compete for the best candidates forever.

The issue in point (iii) is a question of efficiency: 
if a firm's utility is the total quality of all the candidates it hires,
then our measure of social welfare --- the sum of the firms' utilities ---
is simply the total quality of all candidates hired by either of them.
We can then consider the natural performance
guarantee question in this model: how does the social welfare
under multi-period strategic behavior compare to the maximum 
social welfare attainable, where the maximum corresponds to
a central authority that is able 
to impose a matching of candidates to firms?
We obtain a tight bound of 
$\dfrac{2}{1 + \sqrt{1.5}} \approx 0.898$ on the ratio of 
the social welfare under the canonical Nash equilibrium
to the optimal social welfare 
in this model, as the number of periods goes to infinity.
The exact numerical bound here will of course be a property of our modeling 
decisions, but the trade-off that leads to it seems inherent 
in the structure of the multi-period competition.
Moreover, studying this {\em performance ratio}\footnote{We use the neutral
term ``performance ratio'' rather than {\em price of anarchy}
or {\em price of stability} because --- as we will see ---
our game has a natural equilibrium,
and we are interested in the 
relative performance of this natural equilibrium, rather
than necessarily focusing on the best or worst equilibrium.
} 
as a function of the number
of periods, we find that for some settings of the parameters, the
performance ratio is worse for instances with a ``medium'' number
of periods, rather than those with very few (where long-range
planning does not have enough force to favor competition over star
candidates) or those with very many
(where the weaker firm is likely to stop competing, leading to
a higher level of overall employment).

Finally, the issue in point (iv) --- the ``accidental'' effects
of early competition outcomes --- turns out to be analyzable in 
our model via a concrete connection to Polya urn processes 
\cite{pemantle-polya-urn-survey}.
We show how the evolution of the two firms' reputations can be
tracked through an analysis that is closely related to the 
evolving composition of a Polya urn;
however, the analysis is made more complicated by the fact that
the steps in the process are under the control of strategic agents
who are calculating their actions inductively with respect to
the expected outcomes in future periods.
% and the reputations stabilize over time in much the way that
% the composition of the urn does.

This particular combination of phenomena (i)-(iv) appears to be new
from an analytical perspective using formal models, 
despite its familiarity from
everyday experience, and its connections with strands
of more empirical and ethnographic work in economics 
\cite{turban-firm-reputation} and sociology
\cite{coleman-labor-matching,rivera-employer-reputation}.
We thus view the reduced-form
model developed here as correspondingly shedding light on
the interplay between the inherent strategic and probabilistic considerations
as the process unfolds --- including the emergence from the model
of qualitative principles such as the transition between long-running
competition and a decision by one firm to ``give up'' and accept a
lower rank.
Moreover, the use of a model with two firms is consistent with
the long-standing style of analysis in terms of duopolies for
multi-period game-theoretic models (see e.g. 
\cite{beggs-klemperer,benoit-dynamic-duopoly,jun-dynamic-duopoly,padilla-dynamic-duopoly});
two-firm competition is often
the initial place where one looks for principles in establishing 
such models.

As a last point, 
we note that while our model is expressed in terms of job-market
recruiting, there are many settings in which firms compete over
multiple time periods, making decisions that have effects on 
their reputations and hence their relative performance in the future.
As such, the type of probabilistic analysis we carry out here for the 
underlying dynamic multi-period process may be useful in thinking about
the strategic management of evolving reputations more generally ---
thought still of course in the highly reduced form in which we
have expressed it.
For example, it would be interesting to see whether our framework can
be adapted to settings where related issues have been explored,
including the study of product compatibility \cite{chen-compatibility}.
The issue of whether a weaker firm decides to directly compete
with a stronger one, or to avoid direct competition, is also implicit
in studies of the branding and advertising decisions firms make --- 
including whether to explicitly acknowledge a second-place status,
such as the example (discussed in \cite{mullainathan-coarse-thinking})
of the Avis car rental company's ``We Try Harder'' campaign.

\xhdr{Formulating the model}

We now describe the model and its underlying parameters in more detail.
Again, we stress that our model is designed to produce the essential
phenomena in this multi-period competition as cleanly as possible,
and hence is built on two firms that compete in a repeating structure
over multiple periods.
At the same time, the model is set up in such a way that it
can be directly generalized to include larger numbers of firms
and more variability between different time periods.
We discuss some possible extensions briefly in the conclusions
section (Section \ref{sec:conclusions}).

We set up the model as a game with two players over $k$ rounds.
We can think of each player as representing an academic department
that is able to try hiring one new faculty candidate in each of the next
$k$ hiring seasons. 
In each round $t \in \{1, 2, \ldots, k\}$, 
the players are presented with a set of job candidates
with fixed numerical {\em qualities}.
Since we have only two firms in our model, we will assume that the firms'
hiring will only involve considering the two strongest candidates; 
we therefore assume that there are only two candidates available.
Normalizing the quality of the stronger candidate, we define
the qualities of the two candidates to be $1$ and $q < 1$ respectively.

We want to be able to talk separately about a department's 
{\em utility} --- the total quality of all candidates it has
hired --- and its {\em reputation} --- its ability to attract
new candidates based on the quality of the people it has hired.
A number of studies of academic rankings have emphasized that
departments are judged in large part by their strongest members;
intuitively, this is why a smaller department with several ``star''
members can easily rank higher than a much larger department,
and ranking schemas often include measures that focus on this distinction.

Given this, a natural way to define reputation in our model
is to say that the reputation of firm $i$ in round $t$, denoted $x_i(t)$,
is equal to the number of higher-quality candidates (i.e. those
of quality $1$ rather than $q$) that it has hired so far.
This is distinct from the utility of firm $i$ in round $t$, denoted $u_i(t)$,
which is simply the sum of the qualities of all the candidates it has hired.

We assume that a firm is seeking to maximize its utility over
the full $k$ rounds;
{\em however}, note that since this is a multi-period game, 
and reputation determines success in future rounds of hiring,
a firm's equilibrium strategy will in fact involve actions that
are effectively seeking to increase reputation even at the 
expense of short-term sacrifices to expected utility.
This, indeed, is exactly the type of behavior we hope to see
in a model of recruiting.

Building on this discussion, we therefore structure
the game as follows.
\begin{itemize}
\item
Each player $i$ has a numerical {\em reputation} $x_i(t)$
and {\em utility} $u_i(t)$ in round $t$.
We will focus mainly on the case in which the two players each start with
reputation equal to $1$, 
though in places we will consider variations on this initial condition.
\item
In each round $t \in \{1, 2, \ldots, k\}$,
player $i$ chooses one of the candidates $j$ to try recruiting;
this choice of $j$ constitutes the player's {\em strategy} in round $t$.
\item 
{If player $i$ is the only one to try recruiting $j$, then $j$ accepts the offer.
If both players compete for the same candidate $j$, then $j$ accepts player $i$'s offer
with probability proportional to player $i$'s reputation. This follows} the {\em Tullock contest function} that is standard
in economic theory for modeling competition
\cite{skaperdas-contest-success,tullock-contest}, thus we have :
player 1 hires $j$ with probability \\$\dfrac{x_1(t)}{x_1(t) + x_2(t)}$ and
player 2 hires $j$ with probability $\dfrac{x_2(t)}{x_1(t) + x_2(t)}$.
The player who loses this competition for candidate $j$ hires
no one in this round.
%If player $i$ is the only one to try recruiting $j$, then $j$ gets
%hired by $i$.
%If both players compete for the same candidate $j$, then player $i$
%succeeds with probability proportional to its reputation,
%following the {\em Tullock contest function} that is standard
%in economic theory for modeling competition
%\cite{skaperdas-contest-success,tullock-contest}:
%player 1 hires $j$ with probability $\dfrac{x_1(t)}{x_1(t) + x_2(t)}$ and
%player 2 hires $j$ with probability $\dfrac{x_2(t)}{x_1(t) + x_2(t)}$.
%The player who loses this competition for candidate $j$ hires
%no one in this round.
\item 
Finally, each player receives a payoff in round $t$ equal to the
quality of the candidate hired in the round (if any).
The player's utility is increased by the quality of the candidate
it has hired;
the player's reputation is increased by $1$ 
if it has hired the stronger candidate in round $t$,
and remains the same otherwise.
\end{itemize}

Thus the model captures the basic trade-off inherent in recruiting
over multiple rounds --- by competing for a stronger candidate, a player
has the opportunity to increase its reputation by a larger amount,
but it also risks hiring no one.
The model is designed to arrive at
this trade-off using very few underlying parameters;
but we believe that the techniques developed for the analysis
suggest approaches to more complex variants, and we discuss
some of these in the conclusions section (Section \ref{sec:conclusions}).

% (For example, for simplicity we are assuming that if a player
% loses a competition for a candidate in a given
% round, they hire no one in that round;
% it would also be reasonable to consider a more complex model in 
% which a player losing the competition for one candidate could
% go to a back-up option in the same round, 
% though with some probability of failure due to lost time.)

The maximum possible social welfare is achieved if the two players
hire the top two candidates respectively in each round, achieving
a social welfare of $k (1 + q)$.
The key question we consider here is what social welfare can be
achieved in equilibrium for this $k$-round game, and how it compares
to the welfare of the social optimum.
In effect, how much does the struggle for reputation leave
candidates unemployed?

The subgame perfect equilibria in this multi-round game are determined
by backward induction --- essentially, in a given round $t$, a player
evaluates the possible values its utility and
reputation can take in round $t+1$,
after the (potentially probabilistic) outcome of its recruiting in
round $t$.  There are multiple equilibria, but there is a single natural 
class of {\em canonical equilibria} for the model, in which the 
higher-reputation player always goes after the stronger candidate, and ---
predicated on the equilibrium having this form in future rounds ---
the lower-reputation player makes an optimal decision to either compete for
the stronger candidate or make an offer to the weaker candidate.
(When the lower-reputation player is indifferent between these two options,
we break the symmetry using 
the assumption that the lower player hires the weaker candidate.)
The canonical equilibrium can be also viewed as the result of a best response order in which at every round the higher-reputation gets the advantage of making the first choice.
Proving that this structure in fact produces an equilibrium is
non-trivial; in part this is because
reasoning about subgame perfect equilibria always involves
some complexity due to the underlying tree of possibilities,
but the present model adds to this complexity because the randomization
involved in the outcome causes the possible trajectories of the game
to ``explore'' most of this tree.
% \socomment{can we add some sentence saying that proving that this "canonical equilibrium" is in fact an equilibrium is not trivial?}

We study the behavior of this canonical equilibrium, and we define
the {\em performance ratio} of an instance to be the ratio of
total welfare between the canonical 
equilibrium and the social optimum.
%\footnote{We use the neutral
%term ``performance ratio'' rather than {\em price of anarchy}
%or {\em price of stability} because as noted above we are interested in the
%relative performance of the game's natural equilibrium, rather
%than necessarily focusing on its best or worst equilibrium.
%} 

\xhdr{Overview of Results}

We first consider the performance ratio as a function of the number
of rounds $k$.
As an initial question, 
which choice of $k$ yields the worst performance ratio?
When $q < \frac12$, the answer is simple: for $k = 1$,
the players necessarily compete in the one round they have
available, and this yields a performance ratio of $1 / (1 + q)$ ---
as small as possible.
When $q > \frac12$, however, the situation becomes more subtle.
For $k = 1$, the players do not compete in the canonical 
equilibrium, and so the performance ratio for $k = 1$ is $1$.
At the other end of the spectrum, 
when $q \geq \frac12$, the two players will eventually stop
competing with probability $1$ 
and the performance ratio converges up to $1$ when $k$ becomes large.
But in between, the performance ratio can be larger than at
both extremes; in particular, when the quantity
$\dfrac{q}{1-q}$ approaches an integer value $k$ from below,
we show that the performance ratio
is maximized when the number of rounds takes this intermediate value $k$.
% $\left\lceil \dfrac{q}{1-q} \right\rceil$.
These results show how the
time scale over which the players take reputational effects into
account can have a subtle (and in this case non-monotonic) effect
on the efficiency of the job market.

We then turn to the main result of the paper, which is to 
analyze the performance ratio in the limit as the
number of rounds $k$ goes to infinity.
%We then turn to understanding  
%the performance ratio in the limit as the
%number of rounds $k$ goes to infinity.
When $q \geq \frac12$, as just noted, we show that the two players will eventually stop
competing with probability $1$ 
and the performance ratio converges to $1$.
But when $q < \frac12$, something more complex happens:
there is a positive probability, strictly between $0$ and $1$,
that the players compete forever.
This has a natural interpretation --- as reputations evolve,
the two players can settle into relative levels of reputation
under which it is worthwhile for the lower player to compete
for the stronger candidate;
but it may also happen that after a finite number of rounds, one
player decides that it is too weak to continue competing for
the stronger candidate, and it begins to act on its second-tier status.
What is interesting is that each of these outcomes has a
positive probability of occurring.

The possibility of indefinite competition leads to a non-trivial
performance ratio; we show that the worst case occurs when
$q = \sqrt{1.5} - 1 \approx .2247$, with a performance ratio of
$\dfrac{2}{1 + \sqrt{1.5}} \approx 0.898$.
We also show that the performance ratio converges to $1$
as $q$ goes either to $0$ or to $1$.
Our analysis proceeds by defining an urn 
process that tracks the evolution of the players' reputations;
this is a natural connection to develop, since urn processes
are based on models in which probabilities of outcomes in a given
step --- the result of draws from an urn --- are affected by the realized
outcomes of draws in earlier steps.  We provide more background about urn
processes in the next section.
Informally speaking, the fact that a player might compete for
a while and then permanently give up in favor of an alternative option
is also reminiscent of strategies in the multi-armed bandit problem,
where an agent may experiment with a risky option for a while
before permanently giving up and using a safer option;
later in the paper, we make this analogy more precise as well.
To make use of these connections, we study a sequence of games that
begins with players who are constrained to follow a set sequence
of decisions for a long prefix of rounds, and we then successively
relax this constraint until we end up with the original game
in which players are allowed to make strategic decisions from
the very beginning.

In Appendix \ref{sec:fixedp},
we also consider variants of the model in which
one changes the function used for the success probabilities
in the competition between the two players for a candidate.
Note that the way in which competition is handled is an implicit
reflection of the way candidates form preferences over firms based
on their reputations, and hence varying this aspect of the model
allows us to explore different ways in which candidates can behave
in this dimension.
In particular, we consider a variation on the model in which --- when
the two players compete for a candidate --- the lower-reputation 
player succeeds with a fixed probability
$p < \frac12$ and the 
higher-reputation player succeeds with probability $1 - p$.
This model thus captures the long-range competition to become the 
higher-reputation
player using an extremely simple model of competition within each round.
The main result here is that for all $p < q$, the performance 
ratio converges to $1$; the analysis makes use of biased random walks
in place of urn processes to analyze the long-term competition between
the players.

\xhdr{Further Related Work}

As noted above, there has been recent theoretical work studying the
effect of reputation and competition in job markets.
Broder et al. consider hiring strategies designed to increase
the average quality of a firm's employees \cite{broder-lake-wobegon}.
Our focus here is different, due to the feedback effects from
future rounds that our model of competition generates:
a few weak initial hires can make it very difficult for a player
to raise its quality later, while a few strong initial hires can
make the process correspondingly much easier.
Immorlica et al. consider competition between employers, though in 
a quite different model where candidates are presented one at
a time as in the {\em secretary problem} 
\cite{immorlica-dueling-algs,immorlica-competing-secretary},
and each player's goal is to hire a candidate that is stronger
than the competitor's.
They do not incorporate
the spillover of this competition into future rounds.

Our work can also be viewed as developing techniques for
analyzing the performance ratio and/or price of anarchy in
settings that involve dynamic matchings ---
when nodes on one side of a bipartite graph must make 
strategic decisions about matchings to nodes that arrive
dynamically to the other side of the graph.
In the context of job matching, 
Shimer and Smith consider a dynamic matching model of a labor
market in which the central constraint is the cost of
searching for potential partners \cite{shimer-job-search}.
Haeringer and Wooders apply dynamic matching to the problem
of sequential job offers over time \cite{haeringer-job-matching},
but in a setting that considers the sequencing of offers 
in a single hiring cycle;
this leads to different questions, since the consequence
for reputation in future hiring cycles is not in the scope of 
their investigation.
Dynamic matchings have also been appearing in a number of other recent
application contexts 
(e.g. \cite{dickerson-dyn-match-kidneys,zou-dyn-match-manip}),
and there are clearly many unresolved questions here about
the cost of strategic behavior.

\section{The Canonical Equilibrium and its Properties} 
\label{sec:prelim}
An instance of the recruiting game, as described in the introduction,
is defined by the initial reputations $\dep_1$ and $\dep_2$ of the 
two players; the relative quality $\can$ of the weaker
candidate compared to the stronger one; and the number of rounds $k$.
Accordingly, we denote an instance of the game by 
$G_{k,\can}(\dep_1,\dep_2)$. 
Generally $\can$ will be clear from context, and so we will 
also refer to this game as simply $G_{k}(\dep_1,\dep_2)$. 
We will refer to the player of higher reputation as the {\em higher player},
and the player of lower reputation as the {\em lower player}. In case the players
have the same reputation we will refer to player $1$
as the higher player.

The game as defined
is an extensive-form game, and as such it can admit many subgame
perfect equilibria. For example, it is easy to construct a single-round
game in which it is an equilibrium for the lower player to try
to recruit the stronger candidate and for the higher player to go
after the weaker candidate. This equilibrium clearly has a less natural
structure than one in which the higher player goes after
the stronger option; to avoid such pathologies, as noted in the introduction,
we will study multi-round
strategies $s_k(\dep_1,\dep_2)$ that are defined as follows:
%, the canonical
%equilibrium we study is the one in which the higher player
%always makes an offer to the stronger candidate. More formally 

\begin{definition}
Denote by $s_k(\dep_1,\dep_2)$ the following strategies
for the players over the $k$ rounds:
in every round the higher player goes for the stronger candidate and the lower player best-responds by choosing the candidate that maximizes its utility,
taking into account the current round and all later rounds by induction.

For $s_k(\dep_1,\dep_2)$ to be well-defined we make the 
following two assumptions: 
(1) If the lower player is indifferent between going for the stronger candidate and the weaker candidate we assume it chooses to go for the weaker candidate. 
(2) If the two players have the same reputations we break ties in favor of player $1$. 
\end{definition}

The strategies $s_k(\dep_1,\dep_2)$ can be summarized essentially
by saying that in every round of the game, 
first the higher player gets to make an offer to its preferred candidate,
and given this decision 
the lower player makes the choice maximizing its utility. 
To show that the strategies $s_k(\dep_1,\dep_2)$ form a sub-game perfect equilibrium we will show inductively that in every round it is in the higher player's best interest to make an offer to the stronger candidate. 
More formally we denote the strategy of making an offer to 
the stronger candidate in some round by $+$ and 
to the weaker candidate by $-$. We define 
$f(s_k(\dep_1,\dep_2))$ to be the pair of 
strategies that the players use in the first round of $s_k(\dep_1,\dep_2)$. 

We denote player $i$'s utility when the two players play the strategies prescribed by $s_k(\dep_1,\dep_2)$ by 
$u_i(s_k(\dep_1,\dep_2))$. 
We now formally write down the utility of the players 
in $s_k(\dep_1,\dep_2)$ based on the value of $f(s_k(\dep_1,\dep_2))$:
% \socomment{Is it useful to have it here to explain how the players utilities work? or is it confusing?}
% jk: I think it's useful.
\begin{itemize}
\item If $f(s_k(\dep_1,\dep_2)) = \stratpp$ then
\begin{align*}
u_1(s_k(\dep_1,\dep_2)) &= \p{1}{\dep_1}{\dep_2}(1+u_1(s_{k-1}(\dep_1+1,\dep_2))) + \p{2}{\dep_1}{\dep_2}u_1(s_{k-1}(\dep_1,\dep_2+1)) \\ 
u_2(s_k(\dep_1,\dep_2)) &= \p{1}{\dep_1}{\dep_2}u_2(s_{k-1}(\dep_1+1,\dep_2)) +\p{2}{\dep_1}{\dep_2}(1+u_2(s_{k-1}(\dep_1,\dep_2+1))).
\end{align*}
\item If $f(s_k(\dep_1,\dep_2)) = \stratpm$ then
\begin{align*}
u_1(s_k(\dep_1,\dep_2)) &= 1+u_1(s_{k-1}(\dep_1+1,\dep_2)) \\ 
u_2(s_k(\dep_1,\dep_2)) &= \can+u_2(s_{k-1}(\dep_1+1,\dep_2))
\end{align*}
\item If $f(s_k(\dep_1,\dep_2)) = \stratmp$ then
\begin{align*}
u_1(s_k(\dep_1,\dep_2) &= \can+u_1(s_{k-1}(\dep_1,\dep_2+1)) \\ 
u_2(s_k(\dep_1,\dep_2)) &= 1+u_2(s_{k-1}(\dep_1,\dep_2+1))
\end{align*}
\end{itemize}

We denote the social welfare of playing the strategies $s_k(\dep_1,\dep_2)$ by
$$ u(s_k(\dep_1,\dep_2)) =   u_1(s_k(\dep_1,\dep_2))+ u_2(s_k(\dep_1,\dep_2)).$$

Even though it is natural to suspect that the strategies $s_k(\dep_1,\dep_2)$ are indeed a sub-game perfect equilibrium, proving that this is the case is not such a simple task. The first step in showing that the strategies $s_k(\dep_1,\dep_2)$ are a sub-game perfect equilibrium, and a useful fact by itself,
is the monotonicity of the players' utilities $u_i(s_k(\dep_1,\dep_2))$. More formally, 
in Section \ref{sec:app:eq} of the appendix we show that:

\begin{claim} \label{clm:mono}
For any $\dep_1$, $\dep_2$, and $\eps>0$: 
\begin{enumerate}
\item  $u_1(s_k(\dep_1+\eps,\dep_2))\geq u_1(s_k(\dep_1,\dep_2))\geq u_1(s_k(\dep_1,\dep_2+\epsilon))$.
\item  $u_2(s_k(\dep_1,\dep_2+\eps))\geq u_2(s_k(\dep_1,\dep_2))\geq u_2(s_k(\dep_1+\epsilon,\dep_2))$. 
\end{enumerate}
\end{claim}

Next, we prove that the three following statements hold.
\begin{proposition} \label{prop:eq-and-more}
For any integers $\dep_1,\dep_2$ and $k$ the following holds for
the strategies $s_k(\dep_1,\dep_2)$.
\begin{enumerate}
\item $s_k(\dep_1,\dep_2)$ is a sub-game perfect equilibrium in the game $G_k(\dep_1,\dep_2)$.
\item \label{eq:enum:stops-competing} If a player does not compete in the first round of the game $G_k(\dep_1,\dep_2)$, then it does not compete in all subsequent rounds. 
\item \label{eq:enum:stronger} The utility of the higher player in the game 
$G_k(\dep_1,\dep_2)$ is at least as large as
the utility of the lower player.
\end{enumerate}
\end{proposition}
Essentially, we prove all three properties simultaneously by induction
on the number of rounds of the game; to make the inductive argument easier
to follow, we separate the three statements in Subsection \ref{sub:app:eq} of the appendix
into three different claims. Let us mention two more claims that will be useful later on (all proofs are provided in Subsection \ref{sub:app:more_prop} of the appendix):
\begin{claim} \label{claim:kq-non-compete}
If $u_i(s_k(\dep_1,\dep_2))=k\can$, for some player $i$, then player $i$ never competes in the game $G_k(\dep_1,\dep_2)$.
\end{claim}

\begin{claim} \label{clm:keeps_comp}
If player $i$ competes in the first round of the game $G_k(\dep_1,\dep_2)$ and wins, then in the next round of the game it also makes an offer to the stronger candidate.
\end{claim}

\xhdr{Connections to Urn Processes}
Note that since each player's reputation is equal to the number
of stronger candidates it has hired, the reputations are always
integers (assuming they start from integer values).
These integer values evolve while the players are competing;
and once they stop competing, we know by statement (\ref{eq:enum:stops-competing}) of Proposition \ref{prop:eq-and-more} the exact outcome of the game since the players will behave exactly the same as in the game that this is its first round. This brings us to the close resemblance between our recruiting game and a Polya Urn process \cite{pemantle-polya-urn-survey}.
%so

First, let us define what the Polya Urn process is:
\begin{definition}[Polya Urn process]
Consider an urn containing $b$ blue balls and $r$ red balls. The process is defined over discrete rounds. In each round a ball is sampled uniformly at random from the urn; hence  
the probability of drawing a blue ball is $\frac{b}{b+r}$ and
the probability of drawing a red ball is $\frac{r}{b+r}$.
Then, the ball together with another ball of the same color are returned to the urn. 
\end{definition}

There is a clear resemblance between our recruiting game and the urn model. 
As long as the players compete, their reputations evolve in the same way 
as the number of blue and red balls in the urn, 
since the probabilistic
rule for a candidate to select which firm to join is the same as the rule for choosing
which color to add to the urn, and by assumption 
the reputation of the winning player is increased by the 
stronger candidate's quality, which is $1$. 

A striking fact about urn models is that the fraction of the blue (or red) balls converges in distribution as the number of rounds goes to infinity. More specifically, if initially the urn contains a single red ball and a single blue ball then the fraction of blue balls 
converges to a uniform distribution on $[0,1]$ as the number of rounds goes to infinity. More generally, the fraction of blue balls converges to the $\beta$ distribution $\beta(b,r)$. 
Understanding urn processes is useful for understanding our proofs; 
however we should stress that our model and its analysis 
have added complexity due to 
the fact that players stop competing at a point in time that is
strategically determined.

\xhdr{Connections to Bandit Problems}
It is interesting to note that as long as the lower player stays lower our equilibrium selection rule makes this 
effectively a one-player game. In a sense, 
the lower player's strategy in this phase resembles 
the optimal strategy in a mulit-armed bandit problem \cite{git74},
and more specifically in a one-armed bandit problem \cite{Berry-one-arm}. In a one-armed bandit a single player is repeatedly faced with two options (known as ``arms''
following the terminology of slot machines): the player can pull arm $1$,
which gives a reward sampled from some \emph{unknown} distribution, or pull arm $2$ which gives him a reward from a \emph{known} distribution. Informally speaking, by pulling arm $1$ the player gets both a reward and some information about the distribution 
from which the reward is drawn. The player's goal is to maximize its expected reward possibly under some discounting of future rounds. 
A celebrated result establishes
that for some types of discounting (for example geometric) 
one can compute a number called the {\em Gittins index} for each arm (based on one's observations and the prior) 
and the strategy maximizing the player's expected reward is to pull the arm with the highest Gittins index in each round \cite{git74}.
Since by definition the Gittins index of the fixed arm is fixed, this implies that once the Gittins index of the unknown arm drops below the one of the known arm, the player should only pull the known arm. This also means that the player stops collecting information on the distribution of the unknown arm and hence from this round onwards it always chooses the fixed arm. 

There are analogies as well as differences between our game and 
the one-armed bandit problem.  In our game, the lower player is
also faced with a choice between a risky option (competing) and
a safe option (going for the weaker candidate).
On the other hand, 
an important difference between our model and the one-armed bandit problem
is that our game is in fact a two-player game and at any point 
the lower-reputation player can become the higher-reputation one; 
this property contributes additional sources of complexity to the 
analysis of our game.
Moreover, it is
important to note that for many distributions and discount sequences (including the ones most similar to our game) a closed-form expression of the Gittins index is unknown.

%!TEX root =recruiting1c.tex

\section{Analyzing the Game with A Fixed Number of Rounds}
\label{sec:fixedk}
We begin by analyzing
the game played over a fixed number of rounds $k$ and study the
dependence of the performance ratio on $k$.
In the next section, we turn to the main result of the paper,
which is to analyze the limit of the performance ratio as the number of rounds $k$ goes to infinity.

Our first result is a simple but powerful bound  of $\frac{2\can}{1+\can}$ on the performance ratio, which holds for all $k$.
This is done by 
relating the performance ratio to players' decision whether 
to compete in the first round. 
The argument underlying this relationship is quite robust, in that it
is essentially based only on the reasoning that
the players can always decide to stop competing 
and go for the weaker candidate. 
This bound also implies that as $\can$ goes to $1$ the performance ratio also goes to $1$. 

\label{sec:fixed_k}
\begin{claim} \label{clm:poa32}
The performance ratio of any game $G_{k,\can}(\dep_1,\dep_2)$ is at least $\frac{2\can}{1+\can}$.
\end{claim}
\begin{proof}
We begin with the simple observation that the expected social welfare equals the sum of the expected utilities of the two players in the 
beginning of the game. 
To get a lower bound on the performance ratio it is enough to compute an upper
bound on the expected social welfare. This is done by observing that 
$u_i(s_k(\dep_1,\dep_2)) \geq k \can $, since 
a player can always secure a utility of $k \can$ by always making an offer to the weaker candidate. Hence, the following is a bound on the performance ratio: 
$\dfrac{u_1(s_k(\dep_1,\dep_2))+u_2(s_k(\dep_1,\dep_2))}{k(1+\can)} \geq \dfrac{2 k \can}{k(1+\can)} = \dfrac{2\can}{1+\can}.$
\end{proof}
\begin{corollary}
The performance ratio of 
{any game $G_{k,\can}(\dep_1,\dep_2)$} is at least $2/3$. 
\end{corollary}
The previous corollary holds for $\can > 1/2$ since $ \dfrac{2 \can}{1+\can} > 2/3$ 
and for $\can \leq 1/2$ since the performance ratio is trivially lower-bounded 
by $1/(1+\can) \geq 2/3$.

Next, we ask what is the length of a game for which the worst performance ratio is achieved. 
For $\can < 1/2$, this is simply a single-round game. 
However, for $\can > 1/2$ the answer is not so simple. We show that when {$\frac{\can}{1-\can} +\eps$ is an integer for an arbitrarily small $\eps>0$}, 
a game of $k_\can=\frac{\can}{1-\can} +\eps$ rounds exhibits a performance ratio arbitrarily close to $\frac{2\can}{1+\can}$. It is interesting that the players' strategies in the games achieving this maximum performance ratio have a very specific structure -- the players compete just for the first round and then the player who lost goes for the weaker candidate for the rest of the game.

\begin{proposition}
Let $\eps = \lceil \frac{\can}{1-\can}\rceil - \frac{\can}{1-\can}$ and $k_\can =\frac{\can}{1-\can}+\eps$. Then, as $\epsilon$ approaches $0$ from above (remaining strictly positive), 
the performance ratio of the game $G_{k_\can, \can}(\dep,\dep)$ converges to $\frac{2 \can}{1+\can}$.
\end{proposition}
\begin{proof}
Observe that by Claim \ref{clm:first-round} below the players in the game $G_{k_\can, \can}(\dep,\dep)$ compete for the first round (since $\eps>0$) and then completely stop competing. Thus the expected social welfare of the canonical equilibrium is $k+(k-1)\can$ and its performance ratio is:
\begin{align*}
\dfrac{1+(k-1)(1+\can)}{k(1+\can)} = \dfrac{1+ ((\frac{\can}{1-\can}+ \eps)-1)(1+\can)}{(\frac{\can}{1-\can}+ \eps)(1+\can)} = \dfrac{2\can^2+\eps-\eps\can^2}{\can + \can^2+\eps - \eps\can^2}.
\end{align*}
It is not hard to see now that as $\epsilon$ approaches $0$ the performance ratio approaches $\dfrac{2\can}{1+\can} $.
\end{proof}

We now prove for the $k_\can$'s discussed in the previous proposition the players indeed compete only for the first round and then stop competing. More formally we prove:
\begin{claim} \label{clm:first-round} 
In the game $G_{k,\can}(\dep,\dep)$ for $\frac{\can}{1-\can} <  k \leq \frac{1}{1-\can}$ the players compete in the first round and then completely stop competing.
\end{claim}
\begin{proof}
Player $2$ (which is the lower player in the game) competes in the game $G_{k,\can}(\dep,\dep)$ if:
\begin{align*}
\frac{1}{2}(1+u_2(s_{k-1}(\dep,\dep+1)))+\frac{1}{2}u_2(s_{k-1}(\dep+1,\dep)) > \can + u_2(s_{k-1}(\dep+1,\dep)).
\end{align*}
After some rearranging we get that this implies that player $2$ competes if:
\begin{align*}
1+u_2(s_{k-1}(\dep,\dep+1)) > 2\can + u_2(s_{k-1}(\dep+1,\dep)).
\end{align*}
Note that $k \leq \frac{1}{1-\can} \implies \can \geq \frac{k-1}{k}$. Thus, by Claim \ref{clm:2players_split} below we have that for any $\dep_1,\dep_2$ the players in the game $G_{k-1,\can}(\dep_1,\dep_2)$ do not compete. This implies that $u_2(s_{k-1}(\dep,\dep+1)) = k-1$ and $u_2(s_{k-1}(\dep+1,\dep)) = (k-1)\can$. Thus, the players in the game $G_{k,\can}(\dep,\dep)$ compete if $k > (k+1)\can$ implying $\frac{\can}{1-\can} <  k$ as required. 
\end{proof}

Finally we prove:
\begin{claim} \label{clm:2players_split}
If $\can \geq \frac{k}{k+1}$ then  the players in the game $G_{k,\can}(\dep_1,\dep_2)$ never compete.
\end{claim}
\begin{proof}
Let $\dep_2 \leq \dep_1$. Player $2$ competes in the game $G_{k,\can}(\dep_1,\dep_2)$ if:
\begin{align*}
\p{2}{\dep_1}{\dep_2}(1+u_2(s_{k-1}(\dep_1,\dep_2+1))) + \p{1}{\dep_1}{\dep_2}u_2(s_{k-1}(\dep_1+1,\dep_2)) > \can + u_2(s_{k-1}(\dep_1+1,\dep_2)).
\end{align*}
After some rearranging we get that player $2$ competes if:
\begin{align*}
1+u_2(s_{k-1}(\dep_1,\dep_2+1)) > \frac{\dep_1+\dep_2}{\dep_2}\can + u_2(s_{k-1}(\dep_1+1,\dep_2)).
\end{align*}
Observe that $u_2(s_{k-1}(\dep_1,\dep_2+1)) \leq k-1$ as this is the maximum utility a player can get in a $(k-1)$-round game. Also observe that $u_2(s_{k-1}(\dep_1+1,\dep_2)) \geq (k-1)\can$ and that by assumption $\frac{\dep_1+\dep_2}{\dep_2} \geq 2$. Thus, we have that a necessary condition for player $2$ to compete is that $k > (k+1)\can$. This implies that for $\can \geq \frac{k}{k+1}$ player $2$ does not compete in the first round of the game $G_{k,\can}(\dep_1,\dep_2)$. By part (\ref{eq:enum:stops-competing}) of Proposition \ref{prop:eq-and-more} we have that if a player does not compete in the first round of the game it also does not compete in all subsequent rounds which completes the proof. A very similar proof works for the case that player $1$ is the lower player.
\end{proof}

%!TEX root =recruiting1c.tex
\newcommand{\pmf}[3]{f_{#3}(#1,#2)} %1=k, 2=n 3=q
\newcommand{\cdf}[3]{F_{#3}(#1,#2)} %1=k, 2=n 3=q
\newcommand{\tbound} {\frac{4\ln(1/12)}{\ln(1-\ncan)}}
\newcommand{\tboundhoef} { \frac{3\ln(k)-\ln(\can-p)}{(\ncan-p)^2}}
\newcommand{\eswoverk}{1+2\can (\ncan-3\epsilon-\epsilon^2)}
\newcommand{\poabound}{\frac{\eswoverk}{1+\can}}
\newcommand{\ub}{b} 
\newcommand{\kbound}{e^{\frac{8(\ncan-\epsilon)}{\epsilon^3}} + e^{\tbound}}

\section{Analyzing the Long-Game Limit}
\label{sec:performance}
We now turn to the main question in the paper, which
is the behavior of the performance ratio in the limit
as the number of rounds goes to infinity.

Our main result here is
that as $k$ goes to infinity the performance ratio of the 
game $G_k(1,1)$ goes to $\dfrac{1+2 \can \ncan}{1+\can}$, 
where $\ncan = \min\{ \can,\frac{1}{2} \}$. In particular for $\can < 1/2$ this implies that as $k$ goes to infinity the performance ratio goes to $\dfrac{1+2\can^2}{1+\can}$. This function attains its minimum when $\can = \sqrt{1.5}-1 \approx .2247  $ and at this point it has a value of  $\frac{2}{1+\sqrt{1.5} } \approx 0.898$. 
For $\can \geq 1/2$, on the other hand,
this simply implies that as $k$ goes to infinity the performance ratio of the game $G_k(1,1)$ goes to $1$. Defining $\ncan = \min\{ \can,\frac{1}{2} \}$ helps us to present a single unified proof both for $\can<1/2$ and for $\can \geq 1/2$.

The proof of this theorem becomes somewhat involved 
even though its main idea is quite natural. 
Intuitively speaking, we know that as long as the players compete, 
our game proceeds the way an urn process does. 
This means that the probability that player $2$, for example, is the one to hire
the stronger candidate converges to a uniform distribution as the number of rounds $k$ the players compete goes to infinity. 
Henceforth, we will also refer to this probability as player's $2$ \emph{relative reputation}. 
We show that if the relative reputation of one of the players converges to a number smaller than $\ncan$, then
after a fairly 
small number of rounds -- 
specifically $\theta(\ln(k))$ -- 
the players stop competing. 
The probability that the relative reputation of one of the players converges to something less than $\ncan$ is simply $2\ncan$. Therefore, the expected social welfare of our canonical equilibrium converges to $k+2\can \ncan (k-\theta(\ln(k)))$ and the performance ratio converges to $\frac{1+2 \can \ncan}{1+\can}$.

We divide the proof to four subsections. 
In Subsection \ref{subsec:tbind} we introduce {\em $t$-binding games},
which give us a formal way to study games in which
the two players compete for at least the first $t$ rounds. 
By showing that the utilities of the players in our game are at least
as large as their utilities in the $t$-binding game 
we reduce our problem to showing that the expected utility 
in a $t$-binding game is ``large enough''. 
This is done in Subsection \ref{subsec:tesp}. The proof relies on Subsection \ref{subsec:quits} which, loosely speaking, shows that if after $t$ rounds of competition the relative reputation of the lower player 
is non-trivially smaller than $\ncan$ then the lower player stops competing. 
Finally, in Subsection \ref{subsec:wrap} we state the formal theorem and wrap up the proof.
 
 \subsection{$t$-Binding Games}
\label{subsec:tbind}
A recruiting game is {\em $t$-binding}
if in the first $t$ rounds the two players are required to compete for the stronger candidate. We denote a $t$-binding game by $G_k^t(\dep_1,\dep_2)$. We also denote by $s_k^t(\dep_1,\dep_2)$ the canonical equilibrium of the game $G^t_k(\dep_1,\dep_2)$ in which the players compete for the first $t$ rounds and then follow the strategies $s_{k-t}(\dep_1',\dep_2')$ in the resulting game.

Denote by $u(s^t_k(\dep_1,\dep_2))$ the expected social welfare of the canonical equilibrium in the game $G^t_k(\dep_1,\dep_2)$. 
It is intuitive to suspect that making the players compete for the first $t$ rounds can only decrease their utility. In the next lemma we prove that this intuition is indeed correct:
\begin{lemma} \label{lem:t:bind}
The expected social welfare of the game $G_k(1,1)$ is greater than or equal to the expected social welfare of the game $G_k^t(1,1)$; that is,
$u(s_k(1,1))\geq u(s^t_k(1,1))$.
\end{lemma}
\begin{proof}
We prove the lemma by proving a stronger claim:
\begin{claim} 
The expected utility of each of the players in the game $G_k^t(1,1)$ for $~0\leq t < k$ is greater than or equal to their expected utility in the game $G_k^{t+1}(1,1)$.
\end{claim}
\begin{proof}
For simplicity we prove the claim for player $2$; however the claim holds for both players. By definition, in the game $G_k^t(1,1)$ the players compete for at least the first $t$ rounds. 
During this phase of competition, the two players' reputations evolve according
to the update rule for a standard Polya urn process, as described in 
Section \ref{sec:prelim}. A standard result on that process
implies that 
at the end of these $t$ rounds with probability $\frac{1}{t+1}$ player $1$ has a reputation of $1+t-i$ and player $2$ has a reputation of $1+i$ for $0\leq i \leq t$. Thus, we have that:
\begin{align*}
u_2(s_k^t(1,1))=\frac{1}{t+1} \sum_{i=0}^t u_2(s_{k-t}(1+t-i,1+i))
\end{align*}
% \socomment{new proof:}
% jk: this looks good
Let $I_\delta = \{i| f(s_{k-t}(1+t-i,1+i)) = \delta \}$ for $\delta \in \{\langle +,+ \rangle,\langle +,- \rangle,\langle -,+ \rangle  \}$. For example, $I_{\langle +,+ \rangle}$ is the set of all indices $i$ for which the players compete in the first round of the game $G_{k-t}(1+t-i,1+i)$.

We can now write the sum, usefully, as 

\begin{align*}
u_2(s_k^t(1,1)) &=\frac{1}{t+1}\sum_{i \in I_{\langle +,+ \rangle}} u_2(s_{k-t}(1+t-i,1+i)) 
 + \frac{1}{t+1}\sum_{i \in I_{\langle +,- \rangle}} u_2(s_{k-t}(1+t-i,1+i)) \\
 &+ \frac{1}{t+1}\sum_{i \in I_{\langle -,+ \rangle}} u_2(s_{k-t}(1+t-i,1+i)) 
\end{align*}

By this partition:
\begin{itemize}
\item For $i \in I_{\langle +,+ \rangle}$, we have $u_2(s_{k-t}(1+t-i,1+i)) = u_2(s_{k-t}^1(1+t-i,1+i))$ -- since in both of these games the two players compete in the first round.
\item For $i \in I_{\langle +,- \rangle}$, we have $u_2(s_{k-t}(1+t-i,1+i)) \geq u_2(s_{k-t}^1(1+t-i,1+i))$ -- since $u_2(s_{k-t}^{\langle +,- \rangle}(1+t-i,1+i)) \geq u_2(s_{k-t}^{\langle +,+ \rangle}(1+t-i,1+i))$. (in the first round of the game player $2$ prefers going after the weaker candidate over competing).
\item For $i \in I_{\langle -,+ \rangle}$, we have $u_2(s_{k-t}(1+t-i,1+i)) > u_2(s_{k-t}^1(1+t-i,1+i))$ -- since $u_2(s_{k-t}^{\langle -,+ \rangle}(1+t-i,1+i)) = 1+ u_2(s_{k-t-1}(1+t-i+\can,1+i+1)) > u_2(s_{k-t}^{\langle +,+ \rangle}(1+t-i,1+i))$ by monotonicity.
%so
\end{itemize}
%This completes the proof as the utility of player $2$ in the game $G_k^{t+1}(1,1)$ is: 
Thus, we have $u_2(s_k^t(1,1))\geq \frac{1}{t+1} \sum_{i=0}^t u_2(s_{k-t}^1(1+t-i,1+i))=u_2(s_k^{t+1}(1,1))$.
%\begin{align*}
%u_2(s_k^t(1,1))\geq \frac{1}{t+1} \sum_{i=0}^t u_2(s_{k-t}^1(1+t-i,1+i))=u_2(s_k^{t+1}(1,1))
%\end{align*}

%\begin{align*}
%u_2(s_k^{t+1}(1,1)) &=\frac{1}{t+1}\sum_{i \in I_{\langle +,+ \rangle}} u_2(s_{k-t}^1(1+t-i,1+i)) 
% + \frac{1}{t+1}\sum_{i \in I_{\langle +,- \rangle}} u_2(s_{k-t}^1(1+t-i,1+i)) \\
% &+ \frac{1}{t+1}\sum_{i \in I_{\langle -,+ \rangle}} u_2(s_{k-t}^1(1+t-i,1+i)) 
%\end{align*}

\end{proof}
\end{proof}

\subsection{When does the lower player stop competing?}
\label{subsec:quits}
This next phase of our analysis
is composed of two parts: in the first part we show that the utility of the lower player in a $k$-round game is upper bounded by $\max\{ \ub_\can(k,t,\dep), k\can\} $ for some function $b_\can(\cdot)$ to be later defined. In the second part we compute the conditions under which $\ub_\can(k,t,\dep) < k\can$ which implies that under the same conditions the lower player in the game stops competing.

{For this subsection we denote player $1$'s reputation after $t$ rounds by $t-\dep$ and player $2$'s reputation by $\dep$. Both statements below also hold for player $1$ and the game $G_k(x,t-x)$.}

 The following notation will be useful for our proofs:
\begin{itemize}
\item $\pmf{i}{t}{\can} =  {t \choose i} \can^i(1-\can)^{t-i}  $ -- probability mass function for the binomial distribution with $t$ trials.
\item $\cdf{\dep}{t}{\can} = \sum_{i=0}^{\dep} {t \choose i} \can^i(1-\can)^{t-i}$ -- cumulative distribution function for an integer $\dep$.
\end{itemize}

The function that we use to upper bound the player's utility is:
\begin{align*} 
\ub_\can(k,t,\dep) &= \frac{\dep}{t}+3\cdf{\dep}{t}{\ncan} k + (1-3\cdf{\dep}{t}{\ncan} )(k-1)\can \\
&= \frac{\dep}{t}+(k-1)\can+3\cdf{\dep}{t}{\ncan} \cdot \big((k-1)(1-\can)+1 \big)
\end{align*}

To understand the intuition behind the upper bound function $\ub_\can(k,t,\dep)$ it is useful to look at an alternative description of the urn process. 
Under this description, we have a coin whose 
bias is sampled from a uniform distribution on $[0,1]$; then in each round the coin is tossed. 
If the coin turns up heads a blue ball is added to the urn; otherwise a red ball is added to the urn.
Under this alternative description we can think of our lower player as trying to toss this coin (i.e.~competing) in the hope that
its bias is greater than $\ncan$ (recall that $\ncan = \min \{\can, \frac 1 2 \}$). We refer to the event in which the bias of the coin is greater than $\ncan$ as a good event, and the event it is not a bad event. 
To upper-bound the player's utility we assume that if the good event happens the player wins the stronger candidate for all subsequent rounds and hence its utility is $k$. If the bad event happens then the player completely stops competing and thus its utility is $(k-1)\can$.

We show that $\max\{ \ub_\can(k,t,\dep), k\can\}$ is indeed an upper bound on the players' utility as the previous intuition suggests.

\begin{lemma}\label{lem:b-bound}
For any $k$, $\dep$ and $t>\tbound$, 
we have $u_2(s_k(t-\dep,\dep)) \leq \max\{ \ub_\can(k,t,\dep), k\can\}$. 
\end{lemma}
\begin{proof}
We divide the proof into two cases. When, $\ncan \leq \dfrac{\dep+1}{t+1}$ the bound we need to prove is very loose and hence we can prove it directly. However, for $\ncan > \dfrac{\dep+1}{t+1}$ proving this bound is more tricky and for this we use an induction that some times relies on the first case. The proofs of these two cases are \onec{provided in Claim \ref{clm:tull-induction-geq} and Claim \ref {clm:tull-induction} of the appendix.}\twoc{included in the full version.}
 \end{proof}

We can now use the previous bound to compute the conditions under which the lower player prefers to stop competing.
\begin{theorem}\label{thm:stop}
In the game $G_k(t-p\cdot t,p \cdot t)$ for $p = \ncan-\epsilon$, $\eps>0$ and $t = \max\{ \tbound, \tboundhoef \}$  player $2$ does not compete at all.
\end{theorem}
\begin{proof}
By Lemma \ref{lem:b-bound} we have that $u_2(s_k(t-p\cdot t,p \cdot t)) \leq  \max \big \{ \ub_\can(k,t,p\cdot t) , k\can \big \}$ for $t > \tbound$. 
Since we have that $u_2(s_k(t-p\cdot t,p \cdot t)) \geq k\can$, 
if we show that 
$\ub_\can(k,t,p\cdot t) \leq k\can$, 
then we will have $u_2(s_k(t-p \cdot t,p \cdot t)) = k \can $. 
It will then follow from Claim \ref{claim:kq-non-compete} that
the lower player (player $2$) does not compete at all.
The theorem will thus follow if we show that 
for $t = \max\{ \tbound, \tboundhoef \}$,
we have $\ub_\can(k,t,p\cdot t) \leq k\can$.

By Hoeffding's inequality with $\eps = \ncan - p$, we get that $\cdf{p \cdot t}{t}{\ncan} \leq e^{-2t(\ncan-p)^2}$. 
Now, to compute the value of $t$ for which 
$u_2(s_k(t-p \cdot t,p \cdot t)) = k \can $, we simply find 
the value of $t$ for which the following inequality holds:
\begin{align*}
&p + (k-1)\can + 3e^{-2t(\ncan-p)^2}\big((k-1)(1-\can)+1 \big) \leq k \can 
\end{align*}
After some rearranging we get that:
\begin{align*}
&3e^{-2t(\ncan-p)^2 }\big((k-1)(1-\can)+1 \big) \leq \can-p \\
&3(k-1)(1-\can)+1 \leq e^{2t(\ncan-p)^2}(\can-p) 
\end{align*}
Taking natural logarithms we get:
\begin{align*}
&\ln(3(k-1)(1-\can)+1) \leq 2t(\ncan-p)^2 +\ln(\can-p) \\
&\frac{\ln(3(k-1)(1-\can)+1)-\ln(\can-p)}{2(\ncan-p)^2} \leq  t
\end{align*}
In particular this implies that  the claim holds for $t\geq \tboundhoef$.
\end{proof}

\subsection{The Expected Social Welfare of a $t$-Binding Game}
\label{subsec:tesp}
We show that for large enough $k$ the social welfare of the $t$-binding game  $G_k^t(1,1)$ is relatively high. This is done by showing that there exists some $t$, such that after competing for $t$ rounds, with probability $2(\ncan-\epsilon) -\frac{4}{t+1}$ the players reach a game in which the lower player (either player $1$ or player $2$) does not want to compete any more. 
\begin{lemma}\label{lem:t:esw}
For every $\eps>0$ and $k \geq \kbound$, there exists $t$ such that the expected social welfare of the $t$-binding game $G_k^t(1,1)$ is at least 
$k\cdot \left( \eswoverk \right).$
\end{lemma}
\begin{proof}
By the assumption that the game is $t$-binding we have that both players compete over the stronger candidate for the first $t$ rounds. This implies that at the end of these $t$ rounds with probability $\frac{1}{t+1}$ player $1$ has a reputation of $1+t-i$ and player $2$ has a reputation of $1+i$ for $0\leq i \leq t$. Or, in other words, the relative reputation of player $2$ is $\frac{1+i}{t+2}$ with probability $\frac{1}{t+1}$. 

Notice that for any $0 \leq i \leq \lfloor (\ncan-\epsilon) (t+2) \rfloor -2$ it holds that $\frac{1+i}{t+2} < \ncan-\epsilon$. Thus, the probability that the relative reputation of player $2$ is smaller than $(\ncan-\epsilon)$ is 
\begin{align*}
\frac{1}{t+1}  \cdot (\lfloor (\ncan-\epsilon) (t+2) \rfloor -2 +1) 
   \geq \frac{1}{t+1} \cdot ((\ncan-\epsilon) (t+2)  -2) 
   \geq (\ncan-\epsilon) -\frac{2}{t+1}
\end{align*}

This implies that with probability of at least $(\ncan-\epsilon) -\frac{2}{t+1}$ after $t$ rounds the current game is $G_{k-t}\left( (t+2)(1-p), p \cdot (t+2) \right)$ for $p<\ncan-\eps$. Notice that by symmetry the same holds for player $1$. By choosing $t$ that obeys the requirements of Theorem \ref{thm:stop} we get that the lower player in this game does not compete. Therefore, the probability that one of the players stops competing after $t$ rounds is at least $2(\ncan-\epsilon) -\frac{4}{t+1}$. To bound the expected social welfare we make the conservative assumption that with probability $1-(2(\ncan-\epsilon) -\frac{4}{t+1})$ the players compete till the end of the game and get that:
\begin{align*}
u(s^t_k(1,1)) &\geq k + 2\can \left( (\ncan-\epsilon) -\frac{2}{t+1}\right)(k-t) 
\geq k+2\can (\ncan-\epsilon)k - 2\can (\ncan-\epsilon)t -\frac{4k\can}{t}
\end{align*}

%In the \loc, we show how to pick $k$, such that for $t=\frac{2\ln(k)-\ln(\eps)}{\epsilon^2}$ the conditions for both Theorem \ref{thm:stop} and this Lemma hold. \onec{(Claim \ref{clm:cont_binding} )}

Next, we show that for $k \geq \kbound$ and $t=\frac{4\ln(k)-\ln(\eps)}{\epsilon^2}$ the conditions for 
both Theorem \ref{thm:stop} and this Lemma hold. 
%\soedit
{Indeed, if Theorem \ref{thm:stop} holds, 
we have that for every $0<p<r-\epsilon$ the players in the game $G_{k-t}\left( (t+2)(1-p), p \cdot (t+2) \right)$ for $p<\ncan-\eps$ do not compete, as required.}
Recall that Theorem \ref{thm:stop} requires $t+2$ to be at least $\max\{ \tbound, \tboundhoef \}$. Observe that $\frac{4\ln(k)-\ln(\eps)}{\epsilon^2} \geq \tboundhoef$ as by definition $\can-p \geq r-p>\eps$; and that since $\ln(k)> \tbound$ we also have that $t\geq \tbound$.

Next, we show that $u(s^t_k(1,1)) \geq k\cdot \left( \eswoverk \right)$. We begin by plugging in $t=\frac{4\ln(k)-\ln(\eps)}{\epsilon^2}$:
\begin{align*}
u(s^t_k(1,1))&\geq k+2k\can (\ncan-\epsilon) - 2\can (\ncan-\epsilon) \cdot \frac{4\ln(k)-\ln(\epsilon)}{\epsilon^2} -\frac{4k\can}{ \frac{4\ln(k)-\ln(\epsilon)}{\epsilon^2}} \\
&> k+2k\can (\ncan-\epsilon) - 2\can (\ncan-\epsilon) \cdot \frac{4\ln(k)-\ln(\eps)}{\epsilon^2} -\frac{k\can\epsilon^2}{ \ln(k)} \\
&\geq k+2k\can (\ncan-\epsilon-\epsilon^2) - 2\can (\ncan-\epsilon) \cdot \frac{4\ln(k)}{\epsilon^2} +2\can (\ncan-\epsilon)\frac{\ln(\eps)}{\eps^2}
\end{align*}

To prove the Lemma we show that for $k \geq \kbound$ the following two inequalities hold:
%it holds that 1.~$(\ncan-\epsilon) \cdot \frac{2\ln(k)}{k\epsilon^2}<\epsilon$, 2.~$|(\ncan-\epsilon)\frac{\ln(\eps)}{k\eps^2}|<\eps$. %and 3.~ $\frac{2\ln(k)-\ln(\eps)}{\epsilon^2} \geq \tbound$.
\begin{enumerate}
\item $(\ncan-\epsilon) \cdot \frac{8\ln(k)}{k\epsilon^2}<\epsilon$: 
For this we do a variable substitution and denote $\ln (k)=z$, so that
$k=e^z$. Now we find $z$ such that $4(\ncan-\epsilon)z <  \epsilon^3 \cdot e^z$.
By Taylor expansion we have that $e^z > \frac{z^2}{2}$. Thus, we can instead compute when $4(\ncan-\epsilon)z <  \epsilon^3 \cdot \frac{z^2}{2}$ and get that the inequality holds for $z>\frac{8(\ncan-\epsilon)}{\epsilon^3}$. 
This implies that the
inequality holds for $k>e^{\frac{8(\ncan-\epsilon)}{\epsilon^3}}$.
\item $|(\ncan-\epsilon)\frac{\ln(\eps)}{k\eps^2}|<\eps$: This condition also holds for $k>e^{\frac{8(\ncan-\epsilon)}{\epsilon^3}}$ since if $k>e^{\frac{8(\ncan-\epsilon)}{\epsilon^3}}$ by Taylor expansion we have that $k>((\frac{8(\ncan-\epsilon)}{\epsilon^3})^2)/2 = \frac{16(r-\eps)^2}{\eps^6} > \frac{r-\eps}{\eps^3} \cdot |\ln(\eps)|$ and therefore $|(\ncan-\epsilon)\frac{\ln(\eps)}{k\eps^2}|<\eps$.
%\item[3.] This condition is fulfilled by requiring that $k> e^{\tbound}$.
\end{enumerate}
Thus, for  $k\geq\kbound$ and $t=\frac{4\ln(k)-\ln(\eps)}{\epsilon^2}$ we have that $u(s^t_k(1,1)) \geq k\cdot \left( \eswoverk \right)$ as required.
\end{proof}

\subsection{Wrapping up the Proof}
\label{subsec:wrap}
\begin{theorem} \label{thm:per_ratio}
For $\eps>0$ and $k\geq \kbound$, the performance ratio of the game $G_{k}(1,1)$ is at least $\poabound.$
\end{theorem}
\begin{proof}
By Lemma \ref{lem:t:bind} we have that for any $t$, $u(s_k(1,1)) \geq u(s_k^t(1,1))$. By Lemma \ref{lem:t:esw} we have that there exists a $t$ such that $u(s_k^t(1,1)) \geq k \left( \eswoverk \right)$. By combining the two we get that  $u(s_k(1,1)) \geq k \left( \eswoverk \right)$. This means that the performance ratio of the game $G_{k}(1,1)$ is at least  $\frac{k \left( \eswoverk \right)}{k(1+\can)} = \poabound $.
\end{proof}

\begin{corollary}
As $k$ goes to infinity, 
the performance ratio of the game $G_k(1,1,)$ goes to $\frac{1+2\ncan \can}{1+\can}$.
\end{corollary}

%!TEX root =recruiting1c.tex
\section{Conclusions} \label{sec:conclusions}

When firms compete for job applicants over many
hiring cycles, there is a basic strategic tension inherent in the process:
trying to recruit highly sought-after job candidates can build
up a firm's reputation, but it comes with the risk that firm will
fail to hire anyone at all.
In this paper, we have shown how this tension can arise in a
simple dynamic model of job-market matching.
Although our model is highly stylized, it contains a
number of interesting effects that we analyze,
including the way in which
competition can lead to inefficiency through underemployment (quantified in
our analysis of the performance ratio at equilibrium)
and the possibility of different modes of behavior, in which
a weaker firm may end up competing forever, or it may give up
at some point and accept its second-place status.

The model and analysis also suggest a number of directions 
for further investigation.
One direction is to vary
the {\em competition function} that determines the outcome
of a competition between the two firms when they make offers
to the same candidate. As noted above, this can be viewed as
varying the way in which candidates make decisions between
firms based on their reputations.
In Section \ref{sec:fixedp} of the appendix, we explore this issue by considering 
an alternate rule for competition in which the lower-reputation player
wins with a fixed probability $p < \frac12$ (independent of the
difference in reputation) and the higher-reputation player
wins with probability $1 - p$.

This fixed-probability competition function is simpler in structure
than the Tullock function, and it is
illuminating in that it cleanly separates two different aspects of
the strategic decision being made about future rounds.
With the Tullock function, when the lower player competes, it
has the potential for a short-term gain in its success probability
even in the next round (since the ratio of reputations will change),
and it also has the potential for a long-term gain by becoming the
higher player.  With the fixed-probability competition function,
the short-term aspect is effectively eliminated, since as long as
a player remains the lower party, it has the same probability of success;
we are thus able to study strategic behavior about competing when
the only upside is the long-range prospect of becoming the higher player.
We show that the performance is generally much better with this 
fixed-probability rule than with the Tullock function, 
providing us with further insight into
the specific way in which competition leads to inefficiency 
through a reduced performance ratio.

Other directions that lead quickly to interesting questions are
to consider the case of more than two firms, and to consider
models in which the candidates have different characteristics in
different time periods.
For both of these general directions, our initial investigations suggest
that the techniques developed here will be useful for shedding light
on the properties of more complex models that take these issues into account.

\xhdr{Acknowledgments}
We thank 
Itai Ashlagi, Larry Blume, Shahar Dobzinski, Bobby Kleinberg, and Lionel Levine for very
useful suggestions and references.

% \bibliographystyle{plain}
% \bibliography{n}

\newpage

\begin{appendix}
%!TEX root =recruiting1c.tex

 \renewcommand{\p}[3]{ % #1 player index , %2 player 1 strength, #3 player 2 strength
    \IfEqCase{#1}{%
        {1}{c(#2,#3)} %
        {2}{(1-c(#2,#3))} %
        }
        }
 \newcommand{\pd}[3]{ % #1 player index , %2 player 1 strength, #3 player 2 strength
    \IfEqCase{#1}{%
        {1}{c(#2,#3)\cdot} %
        {2}{(1-c(#2,#3))\cdot} %
        }
        }

\section{The canonical equilibrium and its Properties (includes proofs from Section \ref{sec:prelim})} \label{sec:app:eq}
Our main goal in this section is to prove that the strategies 
$s_k(\dep_1,\dep_2)$ form 
a subgame perfect equilibrium in the game $G_k(\dep_1,\dep_2)$, 
and to present the proofs of some useful properties of this equilibrium. 
The arguments can be carried out in a setting more general than that
of the Tullock competition function, and we present them for a broader
class of competition functions, specifying the probability that a
candidate chooses each firm in the event of competition between them.
We work at this greater level of generality for two reasons.
First, it makes clear what properties of the competition function 
are necessary for the equilibrium results.
Second, and more concretely, we study a variant of the model
in Section \ref{sec:fixedp} that involves a different competition function,
in which a candidate picks the lower-reputation firm with a fixed
probability $p < \frac12$ regardless of the actual numerical values of the
reputations.
This fixed-probability
competition function satisfies our more general assumptions, 
and thus we can apply all the results of this section to it.

For ease of exposition, the results in Section \ref{sec:prelim} of
the main text are presented specifically for the Tullock competition function;
as a result, to complete the link back to this section,
we state which claims here generalize each claim 
from Section \ref{sec:prelim}.

Our results hold for the following general definition of a 
competition function $c: \mathbb R \times \mathbb R \rightarrow (0,1)$,
capturing the intuitive notion that $c(x_1,x_2)$ should represent
the probability that player $1$ wins a competition when
the two players' strengths are $x_1$ and $x_2$ respectively.
\begin{definition}
A function $c: \mathbb R \times \mathbb R \rightarrow (0,1)$ is a competition function if:
\begin{itemize}
\item $\p{1}{\dep_1}{\dep_2+\eps} \leq \p{1}{\dep_1}{\dep_2} \leq \p{1}{\dep_1+\eps}{\dep_2}$.
\item For every $\dep_1 \neq \dep_2$: $\p{1}{\dep_1}{\dep_2} = \p{2}{\dep_2}{\dep_1}$.
\item $\p{1}{\dep}{\dep} \geq \p{2}{\dep}{\dep}$.
\end{itemize}
\end{definition}

With this notation in mind the utility of player $2$ for competing is now:
\begin{align*}
\pd{2}{\dep_1}{\dep_2}(1+u_2(s_{k-1}(\dep_1,\dep_2+1))) + \pd{1}{\dep_1}{\dep_2} u_2(s_{k-1}(\dep_1+1,\dep_2)).
\end{align*}

Observe that the following two properties hold for any competition function. These will be useful for later proofs:
\begin{itemize}
\item Let $\eps>0$. $\p{2}{\dep_1}{\dep_2+\eps} \geq \p{1}{\dep_2}{\dep_1}$.
\item Let $\eps>0$. $\p{2}{\dep_2}{\dep_1} \geq \p{1}{\dep_1}{\dep_2+\eps}$.
\end{itemize}

To see why the first statement holds, observe that if $\dep_2 \neq \dep_1+\eps$ we have that:
\begin{align*}
\p{2}{\dep_1}{\dep_2+\eps} = \p{1}{\dep_2+\eps}{\dep_1} \geq \p{1}{\dep_2}{\dep_1}.
\end{align*}
Else, we have that $\dep_2 \neq \dep_1$ and in this case we have that:
\begin{align*}
\p{2}{\dep_1}{\dep_2+\eps} \geq \p{2}{\dep_1}{\dep_2} = \p{1}{\dep_2}{\dep_1}.
\end{align*}
The second statement also holds for similar reasons.

%so: maybe add this to the main text when we define the strategies.
%\begin{observation} \label{obs:sk}
%If $\dep_2 \leq \dep_1$, then by the definition of $s_k(\dep_1,\dep_2)$ we have that
%$u_2(s_k(\dep_1,\dep_2)) = \max\{u_2(s_k^{\langle +,+ \rangle}(\dep_1,\dep_2)), u_2(s_k^{\langle+,-\rangle}(\dep_1,\dep_2)) \}.$ Similarly, if $\dep_1 < \dep_2$, then by the definition of $s_k(\dep_1,\dep_2)$ we have that
%$u_1(s_k(\dep_1,\dep_2)) = \max\{u_1(s_k^{\langle +,+ \rangle}(\dep_1,\dep_2)), u_1(s_k^{\langle-,+\rangle}(\dep_1,\dep_2)) \}.$  \footnote {Note that for now we cannot say anything about the utility of the higher player as it is required by $s_k(\dep_1,\dep_2)$ to go after the stronger candidate.}
%\end{observation}

We begin by showing that since the lower player in $s_k(\dep_1,\dep_2)$ chooses the strategy maximizing its utility it can always guarantee itself a utility of at least $k\can$ by always going for the weaker candidate.
%\begin{claim} \label{clm:guarantee}
%Let player $i$ be the lower player in the game $G_k(\dep_1,\dep_2)$, then, $u_i(s_k(\dep_1,\dep_2)) \geq k \can$.
%\end{claim}
%\begin{proof}
\begin{claim}\label{clm:app:guarantee} 
Let player $i$ be the lower player in the game $G_k(\dep_1,\dep_2)$. Then $u_i(s_k(\dep_1,\dep_2)) \geq k \can$:
\end{claim}
\begin{proof}
We prove the claim for the case that $\dep_2 \leq \dep_1$ but a similar proof can be easily devised for the case that $\dep_1<\dep_2$. We prove the claim by induction on the number of rounds $k$. For the base case $k=1$, it is easy to see that the utility of player $2$ is $\max\{\p{2}{\dep_1}{\dep_2},\can\}$,
and therefore the claim holds. We assume correctness for $(k-1)$-round games and prove for $k$-round games. Observe that:
$u_2(s_k(\dep_1,\dep_2)) = \max\{u_2(s_k^{\langle +,+ \rangle}(\dep_1,\dep_2)), \can + u_2(s_{k-1}(\dep_1+1,\dep_2)) \}.$ Since player $2$ is also the lower player in the game $G_{k-1}(\dep_1+1,\dep_2)$ we can use the induction hypothesis and get that $u_2(s_{k-1}(\dep_1+1,\dep_2)) \geq (k-1) \can$ which completes the proof.
\end{proof}

Next we show that the utilities $u_i(s_k(\dep_1,\dep_2))$ are monotone increasing in player $i$'s reputation and monotone decreasing in its opponent's reputation.
\begin{claim}[generalizes Claim \ref{clm:mono}] \label{clm:app:mono} 
For any $\dep_1$, $\dep_2$, and $\eps>0$: 
\begin{enumerate}
\item  $u_1(s_k(\dep_1+\eps,\dep_2))\geq u_1(s_k(\dep_1,\dep_2))$ and $u_2(s_k(\dep_1+\eps,\dep_2))\leq u_2(s_k(\dep_1,\dep_2))$ 
\item  $u_1(s_k(\dep_1,\dep_2-\eps))\geq u_1(s_k(\dep_1,\dep_2))$ and $u_2(s_k(\dep_1,\dep_2-\eps))\leq u_2(s_k(\dep_1,\dep_2))$. 
\end{enumerate}
\end{claim}
\begin{proof}
We prove both properties concurrently by induction over $k$, the number of rounds in the game. Since the proofs for both properties are very similar, we only present here the proof for the first property. For the base case, $k=1$, we distinguish between the following cases cases:

\begin{enumerate}
%\item \textbf{The players behave the same in $s_1(\dep_1,\dep_2)$ and $s_1(\dep_1+\epsilon,\dep_2)$ :} if they compete 
\item {$s_1(\dep_1,\dep_2)=s_1(\dep_1+\epsilon,\dep_2)$ :} if they compete in
in both $s_1(\dep_1,\dep_2)$ and $s_1(\dep_1+\epsilon,\dep_2)$, then the claim holds simply because $\p{1}{\dep_1+\epsilon}{\dep_2} \geq \p{1}{\dep_1}{\dep_2}$ and $\p{2}{\dep_1+\epsilon}{\dep_2} \leq \p{2}{\dep_1}{\dep_2}$. Else, in both games the players have the exact same utility (either $1$ or $\can$).

\item {$s_1(\dep_1,\dep_2)\neq s_1(\dep_1+\epsilon,\dep_2)$, $s_1(\dep_1,\dep_2)\neq\stratpp$ and $s_1(\dep_1+\epsilon,\dep_2)\neq\stratpp$  :} observe that this is only possible if $\dep_1<\dep_2\leq\dep_1+\epsilon$ and in this case we have that $s_1(\dep_1,\dep_2)=\stratmp$ and $s_1(\dep_1+\epsilon,\dep_2)=\stratpm$ so it is not hard to see that the claim holds.

%\item \textbf{The players compete in $s_1(\dep_1,\dep_2)$ but do not compete in $s_1(\dep_1+\epsilon,\dep_2)$:} this 
\item {$s_1(\dep_1,\dep_2) =\stratpp$ and $s_1(\dep_1+\epsilon,\dep_2)\neq\stratpp$:} this 
implies that $\dep_1 \geq \dep_2$ since the utility player $1$ can get for competing in $G_1(\dep_1+\epsilon,\dep_2)$ is greater than the utility it can get for competing in $G_1(\dep_1,\dep_2)$: $\p{1}{\dep_1+\epsilon}{\dep_2} \geq \p{1}{\dep_1}{\dep_2} > \can$. Therefore, $u_1(s_1(\dep_1+\eps,\dep_2))\geq u_1(s_1(\dep_1,\dep_2))$. For player $2$, $u_2(s_k(\dep_1+\eps,\dep_2))\leq u_2(s_k(\dep_1,\dep_2))$, since we have that $\can <\p{2}{\dep_1}{\dep_2}$.

%\item \textbf{The players do not compete in $s_1(\dep_1,\dep_2)$ but compete in $s_1(\dep_1+\epsilon,\dep_2)$:} this 
\item {$s_1(\dep_1,\dep_2)\neq\stratpp$ and $s_1(\dep_1+\epsilon,\dep_2)=\stratpp$:} this 
implies that $\dep_2>\dep_1$. As for player $2$ its not hard to see that: $u_2(s_1^{\langle +,+ \rangle}(\dep_1,\dep_2)) \geq u_2(s_1^{\langle +,+ \rangle}(\dep_1+\epsilon,\dep_2))$. Thus the only reason that the players do not compete in $s_1(\dep_1,\dep_2)$ is that player $1$ prefers to go for the weaker candidate and it is entitled to make this choice in $s_1(\dep_1,\dep_2)$ only is $\dep_1 < \dep_2$. It is not hard to see that the claim holds for this case as well.
\end{enumerate}

We now assume that both statements $1$ and $2$ hold for $(k-1)$-round games and prove they also hold for $k$-round games. The proof takes a very similar structure to the proof for the base case, except now we shall use the induction hypothesis instead of first principles. The following observation will be useful for the proof:

\begin{observation} \label{obs:compete}
By applying the induction hypothesis we get that the following two statements hold for any $\delta \in \{\stratpp, \stratpm, \stratmp \}$:
\begin{enumerate}
\item $u_1(s_k^\delta(\dep_1+\epsilon,\dep_2)) \geq u_1(s_k^\delta(\dep_1,\dep_2))$
\item $u_2(s_k^\delta(\dep_1,\dep_2)) \geq u_2(s_k^\delta(\dep_1+\epsilon,\dep_2))$
\end{enumerate}
\end{observation}
Take for example the first statement and consider $\delta=\stratpp$. 
To see why it is indeed the case that $u_1(s_k^{\stratpp}(\dep_1+\epsilon,\dep_2)) \geq u_1(s_k^{\stratpp}(\dep_1,\dep_2))$, observe that by the induction hypothesis we have that: $u_1(s_{k-1}(\dep_1+1+\epsilon,\dep_2)) \geq u_1(s_{k-1}(\dep_1+1,\dep_2))$, $u_1(s_{k-1}(\dep_1+\epsilon,\dep_2+1)) \geq u_1(s_{k-1}(\dep_1,\dep_2+1))$, $u_1(s_{k-1}(\dep_1+\eps+1,\dep_2)) \geq u_1(s_{k-1}(\dep_1+\eps,\dep_2+1))$ and that $\p{1}{\dep_1+\epsilon}{\dep_2} \geq \p{1}{\dep_1}{\dep_2}$. Similarly, we can use the induction hypothesis to prove that the two statments are correct for any $\delta \in \{\stratpp, \stratpm, \stratmp \}$.

Just as in the base case, we now distinguish between the following cases:
\begin{enumerate}

\item {$f(s_k(\dep_1,\dep_2)))=f(s_k(\dep_1+\epsilon,\dep_2))$ :} the claim holds by Observation \ref{obs:compete} above.
%\item \textbf{The players behave the same in $f(s_k(\dep_1,\dep_2))$ and $f(s_k(\dep_1+\epsilon,\dep_2))$:} if in both 
%cases they compete then we get that the claim holds by Observation \ref{obs:compete} above. Else, the claim can be %obtained by a simple application of the induction hypothesis. 

\item {$f(s_k(\dep_1,\dep_2))\neq f(s_k(\dep_1+\epsilon,\dep_2))$, $f(s_k(\dep_1,\dep_2))\neq\stratpp$ and $f(s_k(\dep_1+\epsilon,\dep_2))\neq\stratpp$:} observe that this is only possible if $\dep_1<\dep_2\leq\dep_1+\epsilon$ and in this case we have that $f(s_k(\dep_1,\dep_2))=\stratmp$ and $f(s_k(\dep_1+\epsilon,\dep_2))=\stratpm$ so it is not hard to see that the claim holds.

%\item \textbf{The players compete in $f(s_k(\dep_1,\dep_2))$ but do not compete in $f(s_k(\dep_1+\epsilon,\dep_2))$:} 
\item {$f(s_k(\dep_1,\dep_2)) = \stratpp$ and $f(s_k(\dep_1+\epsilon,\dep_2)) \neq \stratpp$:} Similar to the corresponding case for $k=1$, observe that this implies that $\dep_1 \geq \dep_2$. As by Observation \ref{obs:compete} we have that $u_1(s_k^{\langle +,+ \rangle}(\dep_1+\epsilon,\dep_2)) \geq u_1(s_k^{\langle +,+ \rangle}(\dep_1,\dep_2))$. Thus, if the players do not compete in $s_k(\dep_1+\epsilon,\dep_2)$ it can only be because the lower player prefers not to compete and this lower player has to be player $2$. Now, by applying the induction hypothesis for player $1$ we get that $u_1(s_{k-1}(\dep_1+\epsilon +1,\dep_2)) \geq u_1(s_{k-1}(\dep_1+1,\dep_2)) \geq u_1(s_{k-1}(\dep_1,\dep_2+1))$. Thus, it is not hard to see that $u_1(s_k(\dep_1+\eps,\dep_2))\geq u_1(s_k(\dep_1,\dep_2))$.
For player $2$, since as the lower player it chooses to compete in $f(s_k(\dep_1,\dep_2))$ but not in $f(s_k(\dep_1+\epsilon,\dep_2))$ we have that: 
$$u_2(s_k(\dep_1,\dep_2) = u_2(s_k^{\langle +,+ \rangle}(\dep_1,\dep_2)) > u_2(s_k^{\langle+,-\rangle}(\dep_1,\dep_2)) \geq u_2(s_k^{\langle+,-\rangle}(\dep_1+\epsilon,\dep_2)) = u_2(s_k(\dep_1+\epsilon,\dep_2)$$ 
where the last transition is by Observation \ref{obs:compete}. 

\item {$f(s_k(\dep_1,\dep_2)) \neq \stratpp$ and $f(s_k(\dep_1+\epsilon,\dep_2)) = \stratpp$:} this is similar to the previous case only now we have that $\dep_1 < \dep_2$. The reason is that by Observation \ref{obs:compete} we have that $u_2(s_k^{\langle +,+ \rangle}(\dep_1,\dep_2)) \geq u_2(s_k^{\langle +,+ \rangle}(\dep_1+\epsilon,\dep_2))$. Therefore the lower player in the game $G_k(\dep_1,\dep_2)$ is player $1$. After establishing this, it is easy to verify that the claim holds by applying the induction hypothesis in a very similar manner to the previous case. 
\end{enumerate}
\end{proof}

\subsection{$s_k(\dep_1,\dep_2)$ is a Subgame Perfect Equilibrium in the Game $G_k(\dep_1,\dep_2) $} \label{sub:app:eq}
We prove the following three statements simultaneously by induction on the number of rounds in the game:
\begin{proposition}[generalizes Claim \ref{prop:eq-and-more}] \label{prop:app:eq-and-more}
For any integers $\dep_1,\dep_2$ and $k$ the following holds for the strategies $s_k(\dep_1,\dep_2)$.
\begin{enumerate}
\item \label{eq:app} $s_k(\dep_1,\dep_2)$ is a sub-game perfect equilibrium in the game $G_k(\dep_1,\dep_2)$.
\item \label{eq:app:enum:stops-competing} If a player does not compete in the first round of the game $G_k(\dep_1,\dep_2)$, then it does not compete in all subsequent rounds.
\item \label{eq:app:enum:stronger} The utility of the higher player in the game 
$G_k(\dep_1,\dep_2)$ is at least as large as the utility of the lower player.
\end{enumerate}
\end{proposition}
%\soedit
{We separate the simultaneous induction into three numbered statements above.
Proving these three statements by simultaneous induction on $k$
means that in studying properties of
$s_k(\dep_1,\dep_2)$ and $G_k(\dep_1,\dep_2)$, we can 
assume that all three parts hold for $s_{k'}(\dep_1,\dep_2)$ 
and $G_{k'}(\dep_1,\dep_2)$ for every $0<k'\leq k-1$.

We begin by presenting the proof for part \ref{eq:app} of the proposition showing that $s_k(\dep_1,\dep_2)$ is a sub-game perfect equilibrium. As part of the induction the proof relies on the correctness of parts \ref{eq:app:enum:stops-competing} and \ref{eq:app:enum:stronger} for games of less than $k$ rounds. Next, we prove parts \ref{eq:app:enum:stops-competing} and \ref{eq:app:enum:stronger} in Claim \ref{clm:weaker_stops_competing} and Proposition \ref{prp:stronger-wins} respectively. Both proofs assume that $s_{k'}(\dep_1,\dep_2)$ is a subgame perfect equilibrium for every $0<k'\leq k-1$.}

%We begin by presenting the proof that $s_k(\dep_1,\dep_2)$ is a sub-game perfect equilibrium. As part of the induction 
%the proof relies on Claim \ref{clm:weaker_stops_competing} and Proposition \ref{prp:stronger-wins} below, which in turn assume that $s_{k'}(\dep_1,\dep_2)$ is a subgame perfect equilibrium for every $0<k'\leq k-1$. These are also the proofs that once we have shown that $s_k(\dep_1,\dep_2)$ is a subgame perfect equilibrium would show that the first two items hold.
\begin{proposition}\label{prop:app:eq}
The strategies described by $s_k(\dep_1,\dep_2)$ are a subgame perfect equilibrium in the game $G_k(\dep_1,\dep_2)$ for every two integers $\dep_1,\dep_2$.
\end{proposition}
\begin{proof}
We prove the claim by induction on $k$ the number of rounds. We only present the proof for $\dep_1 \geq \dep_2$ as the proof for the case that $\dep_1 < \dep_2$ is very similar. For the base case $k=1$, if $\p{1}{\dep_1}{\dep_2} > \can$, then it is clearly the case that player $1$ competes. Else, since $\dep_1\geq \dep_2$ we have that $\p{2}{\dep_1}{\dep_2} \leq \p{1}{\dep_1}{\dep_2} \leq \can$, therefore player $2$ does not want to compete as well. As clearly player $1$ cannot benefit from competing over the weaker candidate, this implies that there exists an equilibrium in which player $1$ goes for the stronger candidate. 
%Similarly for the case that $\dep_1 < \dep_2$. If $\p{2}{\dep_1}{\dep_2} \geq \can$, then it is clearly the case that player $2$ competes. Else, since $\dep_1 < \dep_2$ we have that $\p{1}{\dep_1}{\dep_2} < \p{2}{\dep_1}{\dep_2} < \can$, therefore player $1$ does not want to compete as well implying that there exists an equilibrium in which $\dep_2$ goes for the stronger candidate. 

Next, we assume correctness for $k'$-round games for any $0<k'\leq k-1$ and prove for $k$-round games. This means we can apply Corollary \ref{cor:weaker_stops_competing} and get that once the lower player prefers to go for the weaker candidate it completely stops competing and Proposition \ref{prp:stronger-wins} to get that the higher player always has greater utility than the lower player.

Assume towards contradiction that in the \emph{unique} equilibrium for the first round of the game $G_k(\dep_1,\dep_2)$ player $1$ goes after the weaker candidate and player $2$ goes after the stronger candidate. Thus, player $1$'s utility is $u_1(s_k^{\langle -,+ \rangle}(\dep_1,\dep_2)) = \can + u_1(s_{k-1}(\dep_1,\dep_2+1))$.
%\begin{align*}
%u_1(s_k^{\langle -,+ \rangle}(\dep_1,\dep_2)) &= \can + u_1(s_{k-1}(\dep_1,\dep_2+1)) 
%\end{align*}

We first observe that by monotonicity if player $1$ prefers to go for the weaker candidate over competing, it has to be the case that $\p{1}{\dep_1}{\dep_2} < \can$. We also observe that it has to be the case that $f(s_k(\dep_1,\dep_2)) =\langle +,+ \rangle$. As it is easy to see that in the case of $f(s_k(\dep_1,\dep_2)) = \langle +,- \rangle$ player $1$ prefers to go for the stronger candidate over competing for the \emph{weaker} candidate.\footnote{Observe that this holds since $\can + u_1(s_{k-1}(\dep_1+\can,\dep_2)) < 1 + u_1(s_{k-1}(\dep_1+1,\dep_2))$ by monotonicity. This is a well defined use of Claim \ref{clm:app:mono} since it does not assume the players' reputations to be integers.}

We now distinguish between $3$ possible scenarios in $s_{k-1}(\dep_1,\dep_2+1)$ which is by the induction hypothesis a subgame perfect equilibrium in $G_{k-1}(\dep_1,\dep_2+1)$:
\begin{enumerate}
\item $f(s_{k-1}(\dep_1,\dep_2+1)) = \langle +,+ \rangle$: observe that in this case player $1$ prefers to compete in the first round of the game $G_k(\dep_1,\dep_2)$ since by monotonicity $u_1(s_{k-2}(\dep_1+1,\dep_2+1)) \geq u_1(s_{k-2}(\dep_1,\dep_2+2))$\footnote{To handle the case of $k=2$ we define $u_i(s_0(\dep_1,\dep_2)) =0$. }; by the induction hypothesis we have that $s_{k-1}(\dep_1,\dep_2+1)$ and $s_{k-1}(\dep_1+1,\dep_2)$ are subgame perfect equilibria we have that $u_1(s_{k-1}(\dep_1+1,\dep_2)) \geq \can +u_1(s_{k-2}(\dep_1+1,\dep_2+1))$ and $u_1(s_{k-1}(\dep_1,\dep_2+1))\geq \can + u_1(s_{k-2}(\dep_1,\dep_2+2))$; and $\p{1}{\dep_1}{\dep_2} \geq \p{1}{\dep_1}{\dep_2+1}$. Thus, we have that: \begin{align*}
u_1(&s_k^{\langle +,+ \rangle}(\dep_1,\dep_2)) =\pd{1}{\dep_1}{\dep_2}(1+u_1(s_{k-1}(\dep_1+1,\dep_2)))+\pd{2}{\dep_1}{\dep_2}u_1(s_{k-1}(\dep_1,\dep_2+1)) \\
 &\geq\pd{1}{\dep_1}{\dep_2}(1+(\can + u_1(s_{k-2}(\dep_1+1,\dep_2+1))))+\pd{2}{\dep_1}{\dep_2}(\can + u_1(s_{k-2}(\dep_1,\dep_2+2))) \\
&\geq \can + \pd{1}{\dep_1}{\dep_2+1}(1+u_1(s_{k-2}(\dep_1+1,\dep_2+1)))+\pd{2}{\dep_1}{\dep_2+1}u_1(s_{k-2}(\dep_1,\dep_2+2)) \\
&= \can + u_1(s_{k-1}^{\langle +,+ \rangle}(\dep_1,\dep_2+1)) = u_1(s_k^{\langle -,+ \rangle}(\dep_1,\dep_2))
\end{align*}
Therefore $\stratmp$ is not an equilibrium for the first round of the game $G_k(\dep_1,\dep_2)$.
\item $f(s_{k-1}(\dep_1,\dep_2+1)) = \langle +,- \rangle$: this implies that $\dep_1 \geq \dep_2+1$. We can apply Corollary \ref{cor:weaker_stops_competing} and get that for player $2$, $u_2(s_{k-1}(\dep_1,\dep_2+1)) = (k-1) \can$. Since by monotonicity we have that $(k-1) \can = u_2(s_{k-1}(\dep_1,\dep_2+1)) \geq u_2(s_{k-1}(\dep_1+1,\dep_2))$ we get that:
\begin{align*}
u_2(s_k^{\langle +,+ \rangle}(\dep_1,\dep_2)) \leq \p{2}{\dep_1}{\dep_2} +(k-1)\can < k\can
\end{align*}
Where the last transition is due to the fact that $\p{2}{\dep_1}{\dep_2} \leq \p{1}{\dep_1}{\dep_2} <\can$. Hence, by Claim \ref{clm:app:guarantee} the lower player (player $2$) does not want to compete in the game $G_k(\dep_1,\dep_2)$.
Therefore $\stratmp$ is not the unique equilibrium for the first round of the game $G_k(\dep_1,\dep_2)$.
%This is in contradiction to Claim \ref{clm:app:guarantee} that for the lower player (player $2$) $u_2(s_k(\dep_1,\dep_2)) \geq k \can$.

\item $f(s_{k-1}(\dep_1,\dep_2+1)) = \langle -,+ \rangle$: This implies that $\dep_1 = \dep_2$ therefore we can apply Corollary \ref{cor:weaker_stops_competing} for player $1$ in the game $G_{k-1}(\dep_1,\dep_2+1)$ and get that $u_1(s_{k}^{\langle -,+ \rangle} (\dep_1,\dep_2)) =\can + u_1(s_{k-1}(\dep_1,\dep_2+1)) = k\can$. We now have the following chain of inequalities, by applying Proposition \ref{prp:stronger-wins}:
\begin{align*}
u_1(s_{k}^{\langle -,+ \rangle} (\dep_1,\dep_2)) > u_1(s_{k}^{\langle +,+ \rangle} (\dep_1,\dep_2)) = u_1(s_{k}(\dep_1,\dep_2)) \geq u_2(s_{k}(\dep_1,\dep_2)) \geq k \can
\end{align*}
The last transition is due to the fact that player $2$ is the lower player in the game $G_k(\dep_1,\dep_2)$ and thus we can use Claim \ref{clm:app:guarantee}. This is of course in contradiction to the fact that $u_1(s_{k}^{\langle -,+ \rangle} (\dep_1,\dep_2)) = k\can$.
\end{enumerate}
\end{proof}

\subsubsection{Proof of Part (\ref{eq:app:enum:stops-competing}) of Proposition \ref{prop:app:eq-and-more}}
We show that if a player prefers not compete in the first round of the game $G_k(\dep_1,\dep_2)$, then it does not compete in all subsequent rounds. This is done by showing that if a player does not compete in the first round of a game, then it does not compete in the second round. The proof assumes that $s_{k-1}(\dep_1,\dep_2)$ is a subgame perfect equilibrium as it used as part of the induction. 
%However, once we have proven that $s_{k}(\dep_1,\dep_2)$ is a subgame perfect equilibrium this assumption is no longer required as it always holds.

Formally we show:
\begin{claim} \label{clm:weaker_stops_competing}
If $s_{k-1}(\dep_1,\dep_2)$ is a subgame perfect equilibrium for every $\dep_1$ and $\dep_2$, then
\begin{itemize}
\item $f(s_k(\dep_1,\dep_2)) = \langle +, - \rangle \implies f(s_{k-1}(\dep_1+1,\dep_2)) = \langle +, - \rangle.$
\item $f(s_k(\dep_1,\dep_2)) = \langle -, + \rangle \implies f(s_{k-1}(\dep_1,\dep_2+1)) = \langle -, + \rangle.$
\end{itemize}
\end{claim}
\begin{proof}
% \socomment{new proof: }
We prove the first statement of the claim as the proof of the second statement is very similar. 
Assume towards a contradiction that $f(s_k(\dep_1,\dep_2)) = \langle +, - \rangle$ but $f(s_{k-1}(\dep_1+1,\dep_2)) \neq \langle +, - \rangle$. The fact that $f(s_k(\dep_1,\dep_2)) = \langle +, - \rangle$ implies that $\dep_2\leq \dep_1$; 
thus if $f(s_{k-1}(\dep_1+1,\dep_2)) \neq \langle +, - \rangle$ it has to be the case that $f(s_{k-1}(\dep_1+1,\dep_2)) = \langle +, + \rangle$. Therefore, player $2$'s utility is:
\begin{align*} %\label{eq:not-comp}
u_2(s_k(\dep_1,\dep_2)) = \can + \pd{2}{\dep_1+1}{\dep_2} (1+u_2(s_{k-2}(\dep_1+1,\dep_2+1)))+\pd{1}{\dep_1+1}{\dep_2} u_2(s_{k-2}(\dep_1+2,\dep_2))
\end{align*}

Observe that the following holds:
\begin{enumerate}
\item $f(s_{k-1}(\dep_1+1,\dep_2)) = \langle +, + \rangle \implies u_2(s_{k-1}(\dep_1+1,\dep_2)) > \can +  u_2(s_{k-2}(\dep_1+2,\dep_2))$.
\item $u_2(s_{k-1}(\dep_1,\dep_2+1)) \geq \can +  u_2(s_{k-2}(\dep_1+1,\dep_2+1))$.
\end{enumerate}
The first of these
statements holds since player $2$ is the lower player in the game $G_{k-1}(\dep_1+1,\dep_2)$,
and the second statement
holds since by assumption $s_{k-1}(\dep_1,\dep_2+1)$ is a subgame perfect equilibrium.
Thus, we have that: 
\begin{align*} %\label{eq:comp}
u_2(s_k^{\langle +,+ \rangle}(\dep_1,\dep_2)) > \can + \pd{2}{\dep_1}{\dep_2} (1+u_2(s_{k-2}(\dep_1+1,\dep_2+1)))+\pd{1}{\dep_1}{\dep_2} u_2(s_{k-2}(\dep_1+2,\dep_2))
\end{align*}
This implies that $u_2(s_k^{\langle +,+ \rangle}(\dep_1,\dep_2)) > u_2(s_k(\dep_1,\dep_2))$ in contradiction to the assumption that player $2$ maximizes its utility by first going for the weaker candidate ($f(s_k(\dep_1,\dep_2)) = \langle +, - \rangle $).
\end{proof}
\begin{corollary} \label{cor:weaker_stops_competing}
For $\dep_2 \leq \dep_1$, if $s_{k'}(\dep_1,\dep_2)$ is a subgame perfect equilibrium for every $0<k'\leq k-1$, $\dep_1$ and $\dep_2$, then, $f(s_k(\dep_1,\dep_2))=\langle +, - \rangle \implies u_2(s_{k}(\dep_1,\dep_2))=k\can$.
\end{corollary}

\subsubsection{Proof of Part (\ref{eq:app:enum:stronger}) of Proposition \ref{prop:app:eq-and-more}}
We show that the utility of the higher player in the game 
$G_k(\dep_1,\dep_2)$ is at least as large as the utility of the lower player. This is based on Claim \ref{clm:player1-leads} and Claim \ref{clm:player2-leads} below.
\begin{proposition} \label{prp:stronger-wins}
If $s_{k-1}(\dep_1,\dep_2)$ is a subgame perfect equilibrium for every two integers $\dep_1$ and $\dep_2$ then:
\begin{itemize}
\item For $\dep_1 \geq \dep_2$: $u_1(s_k(\dep_1,\dep_2)) \geq u_2(s_k(\dep_1,\dep_2))$.
\item For $\dep_2 > \dep_1$: $u_2(s_k(\dep_1,\dep_2)) \geq u_1(s_k(\dep_1,\dep_2))$.
\end{itemize}
\end{proposition}
\begin{proof}
\begin{itemize}
\item For $\dep_1 \geq \dep_2$. By Claim \ref{clm:player1-leads} we have that $u_1(s_k(\dep_1,\dep_2)) \geq u_2(s_k(\dep_2,\dep_1))$. Now, by monotonicity we have that $u_2(s_k(\dep_2,\dep_1)) \geq u_2(s_k(\dep_2,\dep_2))\geq u_2(s_k(\dep_1,\dep_2))$.
\item For $\dep_2 > \dep_1$. By Claim \ref{clm:player2-leads} we have that $u_2(s_k(\dep_1,\dep_2)) \geq u_1(\dep_2-1,\dep_1)$. Now, by monotonicity we have that $u_1(\dep_2-1,\dep_1) \geq u_1(\dep_1,\dep_1)$ since because $\dep_1$ and $\dep_2$ are integers we have that $\dep_2-1\geq \dep_1$. Then we have that $u_1(\dep_1,\dep_1) \geq u_1(\dep_1,\dep_2)$ which completes the proof.
\end{itemize}
\end{proof}

The two following claims allow us to show that the utility of the higher player is always greater by relating between the utilities of player $1$ and player $2$.
\begin{claim} \label{clm:player1-leads}
If for every integers $y',z'$ and $0<k'\leq k-1$ $s_{k'}(y',z')$ is a subgame perfect equilibrium then: $u_1(s_k(y,z)) \geq u_2(s_k(z,y))$ for every $y$ and $z$.
\end{claim}
\begin{proof}
We prove the claim by induction on $k$ the number of rounds. For the base case $k=1$, observe that if $y=z$ then clearly 
$u_1(s_1(y,z)) \geq u_2(s_1(z,y))$; either because the players compete and $\p{1}{y}{y} \geq \p{2}{y}{y}$ or because the players do not compete and $u_1(s_1(y,z)) = 1>\can = u_2(s_1(z,y))$. Else, $y\neq z$, now if $\min\{\p{1}{y}{z},\p{2}{y}{z}\} > \can$, then in both $s_1(y,z)$ and $s_1(z,y)$ the players compete and since $y\neq z$ we have that $\p{1}{y}{z}=\p{2}{y}{z}$, thus $u_1(s_1(y,z)) = u_2(s_1(z,y))$. Else, each of the players goes after a different candidate. If  for example $y>z$, then in both games the player with reputation $y$ goes after the stronger candidate and the player with reputation $z$ goes after the weaker candidate. Thus, $u_1(s_k(y,z)) = u_2(s_k(z,y))$. 

Next, we assume the correctness for $(k-1)$-round games and prove for $k$-round games. We distinguish between the following cases:
\begin{enumerate}
%\item \textbf{The players compete in both $f(s_k(y,z))$ and $f(s_k(z,y))$:} by using the induction hypothesis we get 
\item {$f(s_k(y,z))=f(s_k(z,y))=\stratpp$:} by using the induction hypothesis we get that
that $u_1(s_{k-1}(y+1,z)) \geq u_2(s_{k-1}(z,y+1))$ and $u_1(s_{k-1}(y,z+1)) \geq u_2(s_{k-1}(z+1,y))$. Since $\p{1}{y}{z}\geq\p{2}{y}{z}$ this is sufficient for showing that $u_1(s_k(y,z)) \geq u_2(s_k(z,y))$. More generally this shows that $u_1(s_k^{\langle+,+\rangle}(y,z)) \geq u_2(s_k^{\langle+,+\rangle}(z,y))$.

%\item \textbf{The players do not compete in both $f(s_k(y,z))$ and $f(s_k(z,y))$:} by the definition of $s_k$ we have 
\item {$f(s_k(y,z)) \neq \stratpp$ and $f(s_k(z,y))\neq\stratpp$:} by the definition of $s_k$ we have
three possible subcases:
\begin{itemize}
\item $y<z$: $f(s_k(y,z)) = \langle -,+ \rangle$ and $f(s_k(z,y))= \langle +,- \rangle$.
\item $y=z$: $f(s_k(y,y)) = \langle +,- \rangle$.
\item $y>z$: $f(s_k(y,z)) = \langle +,- \rangle$ and $f(s_k(z,y))= \langle -,+ \rangle$.
\end{itemize}
It is not hard to see that for each one of these subcases we can use the induction hypothesis to show that the claim holds.

%\item \textbf{The players compete in $f(s_k(y,z))$ but not in $f(s_k(z,y))$:} if $z>y$, then:
\item {$f(s_k(y,z))=\stratpp$ and $f(s_k(z,y)) \neq \stratpp$:} if $z>y$, then:
\begin{align*}
u_1(s_k(z,y)) &= 1+u_1(s_{k-1}(z+1,y)) \\
u_2(s_k(y,z)) &= \pd{2}{y}{z} (1+u_2(s_{k-1}(y,z+1))) + \pd{1}{y}{z}u_2(s_{k-1}(y+1,z))
\end{align*}
By using monotonicity and applying the induction hypothesis we get that:
\begin{align*}
u_2(s_{k-1}(y+1,z)) \leq u_2(s_{k-1}(y,z+1)) \leq u_1(s_{k-1}(z+1,y))
\end{align*}
Thus, the claim holds. 

Else, we have that $y>z$. We show that this case is not possible by a \emph{locking argument}. First we observe that this implies that player $1$ is the lower player in the game $G_k(z,y)$ and in the first round of the game it prefers to go for the weaker candidate. Since we assume that for any $y'$,$z'$ and $0<k'\leq k-1$, $s_{k'}(y',z')$ is a subgame perfect equilibrium the requirements of Corollary \ref{cor:weaker_stops_competing} hold and thus we have that: $u_1(s_{k-1}(z,y+1)) = (k-1)\can$. By applying the induction hypothesis we get that $u_2(s_{k-1}(y+1,z)) \leq u_1(s_{k-1}(z,y+1)) = (k-1)\can$. Now, since player $2$ is the lower player in the game $G_k(y+1,z)$ we can apply Claim \ref{clm:app:guarantee} and conclude that $u_2(s_{k-1}(y+1,z))=(k-1)\can = u_1(s_{k-1}(z,y+1))$. Now, the following chain of inequalities provides a contradiction for the assumption that in the game $G_k(z,y)$ the players do not compete as it shows that the lower player (player $1$) actually prefers competing over going for the weaker candidate:
\begin{align*}
u_1(s_k^{\langle +,+ \rangle}(z,y)) \geq u_2(s_k^{\langle +,+ \rangle}(y,z)) > \can + u_2(s_k(y+1,z)) = \can +u_1(s_{k-1}(z,y+1)) = u_1(s_k^{\langle -,+ \rangle}(z,y)).
\end{align*}
The claim is completed since we already treated the case in which $y=z$ as part of the first two cases.
\end{enumerate}\end{proof}

We now prove that a claim similar in spirit to the previous one also holds for player $2$:
\begin{claim} \label{clm:player2-leads} 
If for every integers $y',z'$ and $0<k'\leq k-1$ $s_{k'}(y',z')$ is a subgame perfect equilibrium, then, $u_2(s_k(y,z+1)) \geq u_1(s_k(z,y))$ and $u_2(s_k(z,y))\geq u_1(s_k(y,z+1))$ for every two integers $y$ and $z$.
\end{claim}
\begin{proof}
We prove the two inequalities by induction on $k$ simultaneously. We begin with the base case $k=1$ and distinguish between the following cases: 
\begin{enumerate}
%\item \textbf{The players compete in both $f(s_1(y,z+1))$ and $f(s_1(z,y))$:} in this case $u_2(s_1(y,z+1)) = \p{2}{y}{z
\item {$s_1(y,z+1)=s_1(z,y)=\stratpp$:} in this case $u_2(s_1(y,z+1)) = \p{2}{y}{z+1}$,  $u_1(s_1(z,y)) = \p{1}{z}{y}$, $u_2(s_1(z,y)) = \p{2}{z}{y}$ and $u_1(s_1(y,z+1)) = \p{1}{y}{z+1}$, thus the claim holds. 

%\item \textbf{The players do not compete in both $f(s_1(y,z+1))$ and $f(s_1(z,y))$:} if $y \leq z$ then $u_2(s_1(y,z+1)) 
\item {$s_1(y,z+1) \neq \stratpp$ and $s_1(z,y) \neq \stratpp$:} if $y \leq z$ then $u_2(s_1(y,z+1))
= u_1(s_1(z,y)) =1$ and  $u_1(s_1(y,z+1)) = u_2(s_1(z,y)) =\can$, thus the claim holds. Else, $y \geq z+1$, then $u_2(s_1(y,z+1)) = u_1(s_1(z,y)) =\can$ and $u_1(s_1(y,z+1)) = u_2(s_1(z,y)) =1$. 

\item {The players compete in one of $s_1(y,z+1),s_1(z,y)$ and do not compete in the other:}
Observe that the following hold:
\begin{itemize}
\item If $y\leq z$ then $s_1(y,z+1)=\langle +,+\rangle \implies s_1(z,y) = \langle +,+\rangle$. Observe that $\p{2}{z}{y} \geq \p{1}{y}{z+1}$. This implies that $u_2(s_1^{\langle +,+ \rangle}(z,y)) \geq u_1(s_1^{\langle +,+ \rangle}(y,z+1))>\can$. Now since in both cases the lower player is the player with reputation $y$ the claim follows.
\item If $y\geq z + 1$ then $s_1(z,y) = \langle +,+ \rangle \implies s_1(y,z+1) = \langle +,+ \rangle$. Observe that $\p{2}{y}{z+1} \geq \p{1}{z}{y}$. This implies that $u_2(s_1^{\langle +,+ \rangle}(y,z+1)) \geq  u_1(s_1^{\langle +,+ \rangle}(z,y))>\can$. Now since in both cases the lower player is the player with reputation $z$ or $z+1$ the claim follows.
\end{itemize}
Thus, we are left with the following two sub-cases:
\begin{enumerate}
\item {$s_1(y,z+1) = \stratpp$ and $s_1(z,y) = \stratmp$:} in this case: $u_2(s_1(y,z+1)) = \p{2}{y}{z+1}$, $u_1(s_1(z,y)) = \can$, $u_2(s_1(z,y)) = 1$ and $u_1(s_1(y,z+1)) = \p{1}{y}{z+1}$ and the claim holds.
\item {$s_1(y,z+1) = \stratmp$ and $s_1(z,y) = \stratpp$:} in this case: $u_2(s_1(y,z+1)) = 1$, $u_1(s_1(z,y)) = \p{1}{z}{y}$, $u_2(s_1(z,y)) = \p{2}{z}{y}$ and $u_1(s_1(y,z+1)) = 
\can$ and the claim holds.
\end{enumerate}
\end{enumerate}

%\item \textbf{The players compete in one of $f(s_1(y,z+1)),f(s_1(z,y))$ and do not compete in the other:}
%Observe that the following hold:
%\begin{itemize}
%\item If $y\leq z$ then $s_1(y,z+1)=\langle +,+\rangle \implies s_1(z,y) = \langle +,+\rangle$. Observe that $\p{2}{z}{y} \geq \p{1}{y}{z+1}$. This implies that $u_2(s_1^{\langle +,+ \rangle}(z,y)) \geq u_1(s_1^{\langle +,+ \rangle}(y,z+1))>\can$. Now since in both cases the lower player is the player with reputation $y$ the claim follows.
%\item If $y\geq z + 1$ then $s_1(z,y) = \langle +,+ \rangle \implies s_1(y,z+1) = \langle +,+ \rangle$. Observe that $\p{2}{y}{z+1} \geq \p{1}{z}{y}$. This implies that $u_2(s_1^{\langle +,+ \rangle}(y,z+1)) \geq  u_1(s_1^{\langle +,+ \rangle}(z,y))>\can$. Now since in both cases the lower player is the player with reputation $z$ or $z+1$ the claim follows.
%\end{itemize}
%Thus, we are left with the following two sub-cases:
%\begin{enumerate}
%\item \textbf{The players compete in $f(s_1(y,z+1))$ but not in $f(s_1(z,y))$ and $y\geq z+1$:} in this case: $u_2(s_1(y,z+1)) = \p{2}{y}{z+1}$, $u_1(s_1(z,y)) = \can$, $u_2(s_1(z,y)) = 1$ and $u_1(s_1(y,z+1)) = \p{1}{y}{z+1}$ and the claim holds.
%\item \textbf{The players compete in $f(s_1(z,y))$ but not in  $f(s_1(y,z+1))$ and $y \leq z$:} in this case: $u_2(s_1(y,z+1)) = 1$, $u_1(s_1(z,y)) = \p{1}{z}{y}$, $u_2(s_1(z,y)) = \p{2}{z}{y}$ and $u_1(s_1(y,z+1)) = 
%\can$ and the claim holds.
%\end{enumerate}
%\end{enumerate}

Next, we assume correctness for $(k-1)$-round games and prove for $k$-round games. We distinguish between the same cases as we did for the base case: 
\begin{enumerate}
\item \textbf{$f(s_k(y,z+1))=f(s_k(z,y))=\stratpp$:} the players' utilities are:
%\item \textbf{The players compete in both $f(s_k(y,z+1))$ and $f(s_k(z,y))$:} the players' utilities are:
\begin{align*}
&u_2(s_k(y,z+1)) =  \pd{2}{y}{z+1}(1+u_2(s_{k-1}(y,z+2)))+ \pd{1}{y}{z+1}u_2(s_{k-1}(y+1,z+1)) \\
&u_1(s_k(z,y)) = \pd{1}{z}{y}(1+u_1(s_{k-1}(z+1,y)))+ \pd{2}{z}{y}u_1(s_{k-1}(z,y+1)) \\
&u_2(s_k(z,y)) = \pd{2}{z}{y}(1+u_1(s_{k-1}(z,y+1)))+ \pd{1}{z}{y}u_1(s_{k-1}(z+1,y)) \\
&u_1(s_k(y,z+1)) = \pd{1}{y}{z+1}(1+u_2(s_{k-1}(y+1,z+1)))+ \pd{2}{y}{z+1}u_2(s_{k-1}(y,z+2)) 
\end{align*}
It is not hard to see that by applying the induction hypothesis plus using monotonicity and the facts that
$\p{2}{y}{z+1} \geq \p{1}{z}{y}$ and $\p{2}{z}{y} \geq \p{1}{y}{z+1}$ the claim holds. By this we have actually shown that a stronger statement holds: $u_2(s_k^{\langle +,+ \rangle}(y,z+1)) \geq u_1(s_k^{\langle +,+ \rangle}(z,y))$ and $u_2(s_k^{\langle +,+ \rangle}(z,y)) \geq u_1(s_k^{\langle +,+ \rangle}(y,z+1))$.

%\item \textbf{The players do not compete in both $f(s_k(y,z+1))$ and $f(s_k(z,y))$:} if $y \leq z$ then
\item \textbf{$f(s_k(y,z+1)) \neq \stratpp$ and $f(s_k(z,y))\neq \stratpp$:} if $y \leq z$ then
\begin{align*}
u_2(s_k(y,z+1)) = 1+ u_2(s_{k-1}(y,z+2)) ~&;~ u_1(s_k(z,y)) = 1+ u_1(s_{k-1}(z+1,y)) \\
u_2(s_k(z,y)) = \can+ u_2(s_{k-1}(z+1,y))~&;~ u_1(s_k(y,z+1)) = \can+ u_1(s_{k-1}(y,z+2)).
\end{align*}
Thus we can use the induction hypothesis and get that the claim holds. Else, $y \geq z+1$, and we can again write down the players' utilities and apply the induction hypothesis to get that the claim holds. 
%\begin{align*}
%u_2(s_k(y,z+1)) &= \can+ u_2(s_{k-1}(y,z+1)) \\
%u_1(s_k(z,y)) &= \can +u_1(s_{k-1}(z,y))
%\end{align*}

\item {The players compete in one of $f(s_k(y,z+1)),f(s_k(z,y))$ and do not compete in the other:}
we will show that the following two lemmas hold:
\begin{lemma}
For $y\leq z$, $f(s_k(y,z+1)) = \langle +,+\rangle \implies f(s_k(z,y)) = \langle +,+\rangle$.
\end{lemma}
\begin{proof}
We prove this by using a locking argument very similar to the one we used for Claim \ref{clm:player1-leads}. Observe that in both games the lower player is the player with reputation $y$. Assume towards a contradiction that in $f(s_k(z,y))$ player $2$ does not compete. Since player $2$ is the lower player, we can use Corollary \ref{cor:weaker_stops_competing} to get that $u_2(s_{k-1}(z+1,y))=(k-1) \can$. By applying the induction hypothesis we get that $(k-1)\can = u_2(s_{k-1}(z+1,y)) \geq u_1(s_{k-1}(y,z+2))$. Now, since player $1$ is the lower player in the game $G_{k-1}(y,z+2)$ we can apply Claim \ref{clm:app:guarantee} and conclude that $u_1(s_{k-1}(y,z+2))=(k-1)\can = u_2(s_{k-1}(z+1,y))$. The following chain of inequalities provides a contradiction that in the game $G_k(z,y)$ player $2$ prefers to go for the weaker candidate over competing:
\begin{align*}
u_2(s_k^{\langle +,+ \rangle}(z,y)) \geq u_1(s_k^{\langle +,+ \rangle}(y,z+1)) > \can + u_1(s_k(y,z+2)) = \can +u_2(s_{k-1}(z+1,y)).
\end{align*}
\end{proof}

\begin{lemma}
For $y\geq z + 1$, $f(s_k(z,y)) = \langle +,+\rangle \implies f(s_k(y,z+1)) = \langle +,+\rangle$.
\end{lemma}
\begin{proof}
The proof is very similar to the previous lemma. Observe that in both games the lower player is the player with reputation $z$ or $z+1$. Assume towards a contradiction that in $f(s_k(y,z+1))$ player $2$ does not compete. Since player $2$ is the lower player, we can use Corollary \ref{cor:weaker_stops_competing} to get that $u_2(s_{k-1}(y+1,z+1))=(k-1) \can$. By applying the induction hypothesis we get that $(k-1)\can = u_2(s_{k-1}(y+1,z+1)) \geq u_1(s_{k-1}(z,y+1))$. Now, since player $1$ is the lower player in the game $G_{k-1}(z,y+1)$ we can apply Claim \ref{clm:app:guarantee} and conclude that $u_1(s_{k-1}(z,y+1))=(k-1)\can = u_2(s_{k-1}(y+1,z+1))$. The following chain of inequalities provide a contradiction that in the game $G_k(y,z+1)$ player $2$ prefers to go for the weaker candidate over competing:
\begin{align*}
u_2(s_k^{\langle +,+ \rangle}(y,z+1)) \geq u_1(s_k^{\langle +,+ \rangle}(z,y)) > \can + u_1(s_k(z,y+1)) = \can +u_2(s_{k-1}(y+1,z+1)).
\end{align*}
\end{proof}

Thus, we are left with the following two sub-cases:
\begin{enumerate}
%\item \textbf{The players compete in $f(s_k(y,z+1))$ but not in $f(s_k(z,y))$ and $y\geq z+1$:}
\item {$f(s_k(y,z+1))=\stratpp$ and $f(s_k(z,y))=\stratmp$:}
 observe that by using the induction hypothesis and monotonicity it is not hard to see that the claim holds since:
%\begin{align*}
%u_2(s_k(y,z+1)) &= \p{2}{y}{z+1}(1+u_2(s_{k-1}(y,z+2)))+ \p{1}{y}{z+1}u_2(s_{k-1}(y+1,z+1)) \\
%u_1(s_k(z,y)) &= \can+u_1(s_{k-1}(z,y+1)) 
%\end{align*}
%$u_2(s_k(y,z+1)) \geq u_1(s_k(z,y))$ since 
\begin{align*}
&u_2(s_k(y,z+1)) = u_2(s_k^{\langle +,+ \rangle}(y,z+1) > u_2(s_k^{\langle -,+ \rangle}(y,z+1) \geq u_1(s_k^{\langle +,- \rangle}(z,y) = u_1(s_k(z,y)) \\
&u_2(s_k(z,y)) = u_2(s_k^{\langle +,- \rangle}(z,y)) \geq  u_2(s_k^{\langle +,+ \rangle}(z,y)) \geq  u_1(s_k^{\langle +,+ \rangle}(y,z+1))=u_1(s_k(y,z+1))
\end{align*}
%
%\begin{align*}
%u_2(s_k(z,y)) = u_2(s_k^{\langle +,- \rangle}(z,y)) \geq  u_2(s_k^{\langle +,+ \rangle}(z,y)) \geq  u_1(s_k^{\langle +,+ \rangle}(y,z+1))=u_1(s_k(y,z+1))
%\end{align*}
%
%
%To see that $u_2(s_k(z,y)) \geq u_1(s_k(y,z+1))$ observe that:
%\begin{align*}
%u_2(s_k(z,y)) &= 1+u_2(s_{k-1}(z,y+1)) \\
%u_1(s_k(y,z+1)) &= \p{1}{y}{z+1}(1+u_1(s_{k-1}(y+1,z+1)))+ \p{2}{y}{z+1}u_1(s_{k-1}(y,z+2)) 
%\end{align*}
%it is easy to see that since by the induction hypothesis $u_2(s_{k-1}(z,y+1)) \geq u_1(s_{k-1}(y+1,z+1)) \geq u_1(s_{k-1}(y,z+2)$ the claim holds.

%\item \textbf{The players compete in $f(s_k(z,y))$ but not in  $f(s_k(y,z+1))$ and $y \leq z$:} 
\item {$f(s_k(z,y))=\stratpp$ and $f(s_k(y,z+1))=\stratmp$:} 
observe that by using the induction hypothesis and monotonicity it is not hard to see that the claim holds since:
\begin{align*}
&u_2(s_k(z,y)) = u_2(s_k^{\langle +,+ \rangle}(z,y)) >  u_2(s_k^{\langle +,- \rangle}(z,y)) \geq  u_1(s_k^{\langle -,+ \rangle}(y,z+1))=u_1(s_k(y,z+1)) \\
&u_2(s_k(y,z+1)) = u_2(s_k^{\langle +,- \rangle}(y,z+1) \geq  u_2(s_k^{\langle +,+ \rangle}(y,z+1) \geq u_1(s_k^{\langle +,+ \rangle}(z,y) = u_1(s_k(z,y)) 
\end{align*}

%\begin{align*}
%u_2(s_k(z,y)) = u_2(s_k^{\langle +,+ \rangle}(z,y) \geq u_1(s_k^{\langle +,+ \rangle}(y,z+1) \geq u_1(s_k^{\langle +,- \rangle}(y,z+1) = u_1(s_k(y,z+1))
%\end{align*}
%To see that $u_2(s_k(y,z+1)) \geq u_1(s_k(z,y))$ observe that:
%\begin{align*}
%u_2(s_k(y,z+1)) &= 1+u_2(s_{k-1}(y,z+2)) \\
%u_1(s_k(z,y)) &= \p{1}{z}{y}(1+u_1(s_{k-1}(z+1,y)))+ \p{2}{z}{y}u_1(s_{k-1}(z,y+1)) 
%\end{align*}
%it is easy to see that since by the induction hypothesis $u_2(s_{k-1}(y,z+2)) \geq u_1(s_{k-1}(z+1,y)) \geq u_1(s_{k-1}(z,y+1))$ the claim holds.
\end{enumerate}
\end{enumerate}
\end{proof}

\subsection{Some More Properties of the Canonical Equilibrium} \label{sub:app:more_prop}

We first show that if $u_i(s_k(\dep_1,\dep_2)) = k \can$ then player $i$ goes for the weaker candidate in $f(s_k(\dep_1,\dep_2))$. This immediately implies that if $u_i(s_k(\dep_1,\dep_2)) = k \can$ then player $i$ in $s_k(\dep_1,\dep_2)$ goes for the weaker candidate in every round.
%\begin{claim} \label{clm:kq-not-compete}
%The following two statements hold:
%\begin{itemize}
%\item If $\dep_1<\dep_2$ then $u_1(s_k(x1,x2)) = k\can \implies f(s_k(x1,x2)) = \langle -,+ \rangle$. 
%\item If $\dep_1\geq\dep_2$ then $u_2(s_k(x1,x2)) = k\can \implies f(s_k(x1,x2)) = \langle +,- \rangle$.
%\end{itemize}
%\end{claim}
%\begin{proof}

\begin{claim}[generalization of Claim \ref{claim:kq-non-compete}] \label{clm:app:kq-non-compete}
The following two statements hold:
\begin{itemize}
\item $u_1(s_k(\dep_1,\dep_2)) = k\can \implies f(s_k(\dep_1,\dep_2)) = \langle -,+ \rangle$. 
\item $u_2(s_k(\dep_1,\dep_2)) = k\can \implies f(s_k(\dep_1,\dep_2)) = \langle +,- \rangle$.
\end{itemize}
\end{claim}
\begin{proof}
We prove the claim for player $1$ but a similar proof also works for player $2$. Assume towards a contradiction that $u_1(s_k(\dep_1,\dep_2)) = k\can$ but $f(s_k(\dep_1,\dep_2)) = \langle +,+ \rangle$. By Proposition \ref{prop:app:eq-and-more} we have that $s_k(\dep_1,\dep_2)$ is a subgame perfect equilibrium, thus, if player $1$ prefers to compete it has to be the case that $u_1(s_k^{\langle +,+ \rangle}(\dep_1,\dep_2)) > u_1(s_k^{\langle -,+ \rangle}(\dep_1,\dep_2))$. Observe that $u_1(s_k^{\langle -,+ \rangle}(\dep_1,\dep_2)) \geq k \can$ as a player can always guarantee itself a utility of at least $k\can$ in equilibrium. This is in contradiction to the assumption that $f(s_k(\dep_1,\dep_2)) = \langle +,+ \rangle$.
\end{proof}

%The following two statements hold:
%\begin{itemize}
%\item If $\dep_1<\dep_2$ then $u_1(s_k(x1,x2)) = k\can \implies f(s_k(x1,x2)) = \langle -,+ \rangle$. 
%\item If $\dep_1\geq\dep_2$ then $u_2(s_k(x1,x2)) = k\can \implies f(s_k(x1,x2)) = \langle +,- \rangle$.
%\end{itemize}
%\end{claim}
%\begin{proof}
%We prove the claim for player $1$ but a similar proof also works for player $2$. Assume towards a contradiction that $u_1(s_k(x1,x2)) = k\can$ but $f(s_k(x1,x2)) = \langle +,+ \rangle$. Since $\dep_1<\dep_2$ we have that $u_1(s_k(x1,x2)) = \max \{u_1(s_k^{\langle +,+ \rangle}(\dep_1,\dep_2)), u_1(s_k^{\langle -,+ \rangle}(\dep_1,\dep_2)) \}$. Since we assume that in $f(s_k(x1,x2))$ a player competes only if it can strictly benefit from this, we get the following contradiction:
%\begin{align*}
%k \can = u_1(s_k^{\langle +,+ \rangle}(\dep_1,\dep_2)) > u_1(s_k^{\langle -,+ \rangle}(\dep_1,\dep_2)) \geq k \can
%\end{align*} 
%where the last inequality is due to Claim \ref{clm:app:guarantee}.
%\end{proof}

Next, based on Claim \ref{clm:weaker_stops_competing} we can show that if a player competes and wins then in the next round it prefers competing over going for the weaker candidate.

\begin{claim}[generalization of Claim \ref{clm:keeps_comp}] \label{clm:app:keeps_comp}
If $f(s_k(\dep_1,\dep_2)) = \langle+,+ \rangle$ then:
\begin{itemize}
\item $f(s_{k-1}(\dep_1+1,\dep_2)) \in \{ \langle+,+ \rangle,  \langle+,- \rangle\}$
\item $f(s_{k-1}(\dep_1,\dep_2+1)) \in \{ \langle+,+ \rangle,  \langle -,+\rangle \}$
\end{itemize}
\end{claim}
\begin{proof}
We first show that $f(s_{k-1}(\dep_1+1,\dep_2)) \in \{ \langle+,+ \rangle,  \langle+,- \rangle\}$. Assume towards contradiction that $f(s_{k-1}(\dep_1+1,\dep_2))=\langle-,+ \rangle$. First observe that if $\p{1}{\dep_1}{\dep_2} >\can$ then $\p{1}{\dep_1+1}{\dep_2} >\can$ thus player $1$ maximizes its utility by competing in the next round as well. It also has to be the case that $\dep_1+1<\dep_2$ since otherwise as the higher player in the game $G_k(\dep_1+1,\dep_2) $ player $1$ should go for the stronger candidate. By Corollary \ref{cor:weaker_stops_competing} we have that $u_1(s_{k-1}(\dep_1+1,\dep_2))=(k-1)\can$. By monotonicity we have that $u_1(s_{k-1}(\dep_1,\dep_2+1)) \leq u_1(s_{k-1}(\dep_1+1,\dep_2)) = (k-1)\can$. Thus, we have that $u_1(s_k(\dep_1,\dep_2)) \leq \p{1}{\dep_1}{\dep_2} +(k-1)\can \leq k\can$. This implies by Claim \ref{clm:app:kq-non-compete} that player $1$ does not compete in $f(s_k(\dep_1,\dep_2))$ in contradiction to the assumption. The proof of the second statement regarding player $2$ is very similar and hence omitted. 
\end{proof}

%!TEX root =recruiting1c.tex

\section{Proofs from Section 3}

\begin{claim}\label{clm:tull-induction-geq}
For any $k$, $\dep$ and $t>\tbound$ such that $\ncan \leq \dfrac{\dep+1}{t+1}$,  $u_2(s_k(t-\dep,\dep)) \leq \max\{ \ub_\can(k,t,\dep), k\can\}$
\end{claim}
\begin{proof}
We actually prove a stronger claim,
which is that for $t>\tbound$, $ \ub_\can(k,t,\dep) \geq  k$. 
The utility of any player in a game of $k$ rounds is at most $k$ 
and hence this will be enough to complete the proof. 
To do this we need to show that $3\cdf{\dep}{t}{\ncan} \geq 1$ for $t>\tbound$. By \cite{hamza-binomial-median} we have that the median of a binomial distribution is at distance of at most $\ln(2)$ from its mean. Thus, if $\dep \geq \ncan t +\ln(2)$ we are done, as in this case $3\cdf{\dep}{t}{\ncan} \geq \frac 3 2$. Else, $\ncan t+\ncan-1 \leq\dep < \ncan t +\ln(2)$. This implies that $\dep+2\geq \ncan t +\ln(2)$, which in turn implies that $\cdf{\dep+2}{t}{\ncan} \geq \frac{1}{2}$. Since $\cdf{\dep+2}{t}{\ncan} = \cdf{\dep}{t}{\ncan} +\pmf{\dep+1}{t}{\ncan}+\pmf{\dep+2}{t}{\ncan}$, what left to show is that $\pmf{\dep+1}{t}{\ncan}+\pmf{\dep+2}{t}{\ncan} < \frac{1}{6}$:
\begin{align*}
\pmf{\dep+i}{t}{\ncan} &={t \choose \dep+i} \ncan^{\dep+i}(1-\ncan)^{t-(\dep+i)} 
\leq  \big(\dfrac{x+i}{t} \big)^{-(x+i)} \ncan^{\dep+i}(1-\ncan)^{t-\dep-i} \\
&\leq \ncan^{-(x+i)} \ncan^{\dep+i}(1-\ncan)^{t-\dep-i} 
= (1-\ncan)^{t-\dep-i} 
\leq (1-\ncan)^{\frac{1}{2}t-3}
\end{align*}
The last transition is due to the fact that $\dep < \ncan t +\ln(2)$, $\ncan \leq 1/2$ and $i\leq 2$. Thus we have that $\pmf{\dep+1}{t}{\ncan}+\pmf{\dep+2}{t}{\ncan} \leq 2(1-\ncan)^{\frac{1}{2}t-3}$. To compute when $2(1-\ncan)^{\frac{1}{2}t-3} < \frac{1}{6}$ we take a natural logarithm and get that: $(\frac{1}{2}t-3)\cdot \ln(1-\ncan)< \ln({\frac{1}{12}})$, hence it is not hard to see that the claim holds for $t> \tbound$. 
\end{proof}

\begin{claim} \label{clm:tull-induction}
For any $k$, $\dep$ and $t>\tbound$ such that $\ncan > \dfrac{\dep+1}{t+1}$,  $u_2(s_k(t-\dep,\dep)) \leq \max\{ \ub_\can(k,t,\dep), k\can\}$
\end{claim}
\begin{proof}
Recall that $\ub_\can(k,t,\dep) = \frac{\dep}{t}+(k-1)\can+3\cdf{\dep}{t}{\ncan} \cdot \big((k-1)(1-\can)+1 \big)$.
We prove the claim by induction. We first observe that the claim holds for $k=1$. Notice that by the assumption that $\frac{\dep+1}{t+1} < \ncan \leq \frac{1}{2}$ we have that player $2$ is the lower player. Thus, $u_2(s_1(t-\dep,\dep)) \leq \max \big \{ \frac{\dep}{t}, \can \big \}$ and the claim holds. Next, we assume correctness for $k-1$ rounds and prove for $k$. If $u_2(s_k(t-\dep,\dep))=k\can$, then the induction hypothesis holds and we are done. Otherwise, we have that $u_2(s_k(t-\dep,\dep)) > k\can$, this immediately implies that $f(s_k(t-\dep,\dep)) = \langle +,+ \rangle$.

By Claim \ref{clm:keeps_comp} this implies that $u_2(s_{k-1}(t-\dep,\dep+1))> (k-1)\can$. Hence, either by the induction hypothesis (if $\ncan>\frac{\dep+2}{t+2} $) or by Claim \ref{clm:tull-induction-geq} (if $\ncan \leq \frac{\dep+2}{t+2}$) we have that $u_2(s_{k-1}(t-\dep,\dep+1)) \leq \ub_\can (k-1,t+1,\dep+1)$. Since $\frac{\dep+1}{t+2} < \ncan$. We can also use the induction hypothesis to get that $u_2(s_{k-1}(t-\dep+1,\dep)) \leq \max\{ \ub_\can(k-1,t+1,\dep), k\can\}$.

%For $u_2(s_{k-1}(t-\dep+1,\dep))$ we can use the induction hypothesis as $\frac{\dep+1}{t+2} < \ncan$ and get that:
%\begin{align*}
%u_2(s_{k-1}(t-\dep+1,\dep)) \leq \max\{ \ub_\can(k-1,t+1,\dep), k\can\}
%\end{align*}

To show that $u_2(s_k(t-\dep,\dep)) \leq \max\{ \ub_\can(k,t,\dep), k\can\}$ we now distinguish between two cases:
\begin{enumerate}
\item $u_2(s_{k-1}(t-\dep+1,\dep)) \leq  \ub_\can(k-1,t+1,\dep)$:
\begin{align*}
u_2&(s_k(t-\dep,\dep)) \leq \frac{\dep}{t} \left( 1+ \ub_\can(k-1,t+1,\dep+1)\right ) 
    + \frac{t-\dep}{t} \ub_\can(k-1,t+1,\dep) \\
%&\leq \frac{\dep}{t} \left( 1+ \frac{\dep+1}{t+1}+3\cdf{\dep+1}{t+1}{\ncan} \cdot \big((k-2)(1-\can)+1 \big) + (k-2)\can \right ) \\
%&+ \frac{t-\dep}{t} \left( \frac{\dep}{t+1}+3\cdf{\dep}{t+1}{\ncan} \cdot \big((k-2)(1-\can)+1 \big) + (k-2)\can \right) \\
&= \frac{\dep}{t} + \frac{\dep}{t} \cdot \frac{\dep+1}{t+1} + \frac{t-\dep}{t}\cdot\frac{\dep}{t+1}+(k-2)\can 
    + 3\cdf{\dep}{t+1}{\ncan} \cdot \big((k-2)(1-\can)+1 \big) \\
&~~~+ 3\frac{x}{t} \cdot \pmf{\dep+1}{t+1}{\ncan} \cdot \big((k-2)(1-\can)+1 \big) \\
&\leq^{(1)} \frac{2\dep}{t} +(k-2)\can + 3\cdf{\dep}{t}{\ncan} \cdot \big((k-2)(1-\can)+1 \big) \\
&\leq^{(2)} \frac{\dep}{t}+ (k-1)\can + 3\cdf{\dep}{t}{\ncan} \cdot \big((k-1)(1-\can)+1 \big) = \ub_\can(k,t,\dep)
\end{align*}
Transition $(1)$ is obtained by applying Claim \ref{clm:comb} (below) and some rearranging. For transition $(2)$ we use the fact that $\frac{\dep}{t} <\frac{\dep+1}{t+1} < \ncan \leq \can$.

\item $u_2(s_{k-1}(t-\dep+1,\dep)) = (k-1)\can$:
\begin{align*}
u_2&(s_k(t-\dep,\dep)) \leq \frac{\dep}{t} \left( 1+ \ub_\can(k-1,t+1,\dep+1)\right ) 
%u_2(s_k(t-\dep,\dep))  &\leq  \frac{\dep}{t} \left( 1+ \frac{\dep+1}{t+1}+3\cdf{\dep+1}{t+1}{\ncan}\cdot \big((k-2)(1-\can)+1 \big) + (k-2)\can \right )\\ 
   +\frac{t-\dep}{t} (k-1) \can \\
&=  \frac{\dep}{t} + \frac{\dep}{t} \cdot \frac{\dep+1}{t+1} + \frac{\dep}{t}(k-2)\can +  \frac{\dep}{t} \cdot 3\cdf{\dep+1}{t+1}{\ncan} \cdot \big((k-2)(1-\can)+1 \big) +\frac{t-\dep}{t}(k-1)\can    
    \\
&=  \frac{\dep}{t} + \frac{\dep}{t} \cdot \frac{\dep+1}{t+1}  +\frac{t-\dep}{t}\can + (k-2)\can    
   +  \frac{\dep}{t} \cdot 3\cdf{\dep+1}{t+1}{\ncan} \cdot \big((k-2)(1-\can)+1 \big) \\
 &\leq  \frac{\dep}{t} + (k-1)\can 
     +  \frac{\dep}{t} \cdot 3\cdf{\dep+1}{t+1}{\ncan}\cdot \big((k-2)(1-\can)+1 \big)
\end{align*}
For the last transition we use the fact that $\frac{\dep}{t} \cdot \frac{\dep+1}{t+1} < \frac{\dep}{t} \cdot \can$ since by assumption we have that $\frac{\dep+1}{t+1} < \ncan \leq \can $. 

Notice that $\frac{\dep}{t} \cdf{\dep+1}{t+1}{\ncan}  < \cdf{\dep}{t+1}{\ncan} + \frac{\dep}{t} \pmf{\dep+1}{t+1}{\ncan}$.
Hence by applying Claim \ref{clm:comb} (below) we get that $\frac{\dep}{t} \cdf{\dep+1}{t+1}{\ncan}  < \cdf{\dep}{t}{\ncan}$ which completes the proof.
\end{enumerate}
~~
\end{proof}

\begin{claim} \label{clm:comb}
If $\dep <t$ then, $ \cdf{\dep}{t+1}{\can} + \frac{\dep}{t} \pmf{\dep+1}{t+1}{\can} \leq \cdf{\dep}{t}{\can}$
\end{claim}
\begin{proof}
\begin{align*}
 \cdf{\dep}{t+1}{\ncan} &= \sum_{i=0}^{\dep} {t+1 \choose i} \can^i(1-\can)^{t+1-i} \\
 &= (1-\can) \sum_{i=0}^{\dep}\frac{t+1}{t+1-i} {t \choose i} \can^i(1-\can)^{t-i} \\
 &= (1-\can) \sum_{i=0}^{\dep}(1+\frac{i}{t+1-i}) {t \choose i} \can^i(1-\can)^{t-i} \\
 &= (1-\can) \cdf{\dep}{t}{\can} + (1-\can)\sum_{i=1}^{\dep}\frac{i}{t+1-i} \cdot \frac{t!}{i!(t-i)!} \can^i(1-\can)^{t-i} \\
 &= (1-\can) \cdf{\dep}{t}{\can} +  (1-\can)\sum_{i=1}^{\dep}\frac{t!}{(i-1)!(t+1-i)!} \can^i(1-\can)^{t-i} \\
%  &= (1-\can) \cdf{\dep}{t}{\can} + \\ & ~~~~~~~~ (1-\can)\sum_{i=1}^{\dep}\frac{i}{t+1-i} \cdot \frac{t!}{i!(t-i)!} \can^i(1-\can)^{t-i} \\
 %   &= (1-\can) \cdf{\dep}{t}{\can} + \\ & ~~~~~~~~ (1-\can)\sum_{i=1}^{\dep}\frac{t!}{(i-1)!(t+1-i)!} \can^i(1-\can)^{t-i} \\
  &= (1-\can) \cdf{\dep}{t}{\can} + (1-\can)\sum_{i=1}^{\dep} {t \choose i-1} \can^i(1-\can)^{t-i} \\
  &= (1-\can) \cdf{\dep}{t}{\can} + \can\sum_{i=1}^{\dep} {t \choose i-1} \can^{i-1}(1-\can)^{t-i+1} \\
  &= (1-\can) \cdf{\dep}{t}{\can} + \can\sum_{i=0}^{\dep-1} {t \choose i} \can^{i}(1-\can)^{t-i} \\
  &= (1-\can) \cdf{\dep}{t}{\can} + \can \cdot \cdf{\dep-1}{t}{\can} \\
  &=  \cdf{\dep}{t}{\can} - \can  \cdot \pmf{\dep}{t}{\can}
\end{align*}
It remains to show that $\can \pmf{\dep}{t}{\can} > \frac{\dep}{t}{\can} \pmf{\dep+1}{t+1}{\can}$ which is done by noticing that since $\dep < t $ then:
\begin{align*}
 \frac{\dep}{t} \pmf{\dep+1}{t+1}{\can} &= \frac{\dep}{t} {t+1 \choose \dep+1} \can^{\dep+1}(1-\can)^{t-\dep}  \\
 & =  \frac{\dep}{t} \cdot \frac{t+1}{\dep+1} \cdot \can \cdot {t \choose \dep} \can^{\dep}(1-\can)^{t-\dep}  \\
  &= \frac{\dep}{t} \cdot \frac{t+1}{\dep+1} \cdot \can \cdot \pmf{\dep}{t}{\can} \\
 &< \can \cdot \pmf{\dep}{t}{\can}
 \end{align*}
~~
\end{proof}

%!TEX root =recruiting1c.tex
\section{Other Competition Functions: Fixed Probability} \label{sec:fixedp}
One of the key components of our model is the underlying
{\em competition function}: when players of reputation $\dep_1$ and $\dep_2$
respectively compete for the same candidate in a given round,
the competition function specifies the probability that the candidate
selects each player, in terms of $\dep_1$ and $\dep_2$.
In this section we explore the effect that using other competition functions 
has on the performance ratio. 
An extreme example is when the higher player deterministically wins the competition (and if both players have the same reputation, then each wins with probability $1/2$). Using this competition rule in the game $G_k(\dep,\dep)$ the players only compete for the first round to ``discover'' who is the higher player and then stop competing. Thus the performance ratio of this game is very close to $1$. In this section we study a natural generalization of this function.  

Consider a competition function specifying that the lower player
wins with a fixed probability $p < 1/2$, and the higher player wins 
with probability $(1-p)$.
In case the two players have the same reputations, ties are broken in favor of player $1$. 
Clearly if $p>\can$, then the players compete forever, since the lower player gains more from competing than 
from going for the weaker candidate. 
Therefore, we assume from now on that $p<\can$.
We observe that this competition function belongs to the set of competition functions defined in Section \ref{sec:app:eq} of the appendix and hence we can make use of all claims specified there. For example, we have that the strategies $s_k(\dep_1,\dep_2)$ form a subgame perfect Nash equilibrium in this game (Proposition \ref{prop:app:eq} ), the players' utilities are monotone (Claim \ref{clm:app:mono}) and that once a player decides to go after the lower candidate it will do so in all subsequent rounds (Claim \ref{clm:weaker_stops_competing}).  

We first show that once the absolute value of the difference 
between the players' reputations reaches $\theta(\log(k))$, 
the lower player stops competing (Claim \ref{clm:frac_utility} and Lemma \ref {lem:fixedp_drops}). 
Then, in Lemma \ref{lem:fixedp_exp} we show  that the expected number of rounds it takes the players to reach such a difference in reputations, starting from equal reputations, is also $\theta(\log(k))$. This implies that as $k$ goes to infinity the performance ratio goes to $1$, as we prove in Theorem \ref{thm:fixedp_poa}.

Our first step is similar in spirit to the proof for the 
Tullock competition function; we show that the utility of the lower player is bounded by $\max \{\ub_\can^p(k,d), k \can\}$, where
\begin{align*}
\ub_\can^p(k,d) &= p + (\frac{p}{1-p})^d k +(1-(\frac{p}{1-p})^d)(k-1)\can \\
&= p+(\frac{p}{1-p})^d \big((k-1)(1-\can)+1 \big) + (k-1)\can
\end{align*}

This is obtained by induction over the
difference in the reputations of the two players.
The intuition for the upper bound function is also similar. 
In the good event, the lower player becomes the higher player and wins all subsequent rounds; hence its utility is $k$. 
In the bad event the lower player stays the lower player and loses the reward for competing this round; hence its utility is at most $(k-1)\can$. 
To compute the probability of the good event, we can imagine that $d$ (the difference between the players' reputations) is the initial location of a particle performing a biased random walk that goes left with probability $p$ and right with probability $1-p$. 
Under this view, the probability that this particle ever reaches $0$ ---
and hence that the difference $d$ ever reaches $0$ --- is
$\frac{p}{1-p}$ \cite{feller-vol1}. 
We formalize this intuition in the next claim. The claim is stated and proved for player $2$ but a similar claim also holds for player $1$.
\begin{claim} \label{clm:frac_utility}
For any $d\geq 0$ and any $k$: $u_2(s_k(\dep,\dep-d)) \leq \max \{\ub_\can^p(k,d), k \can\}$.
\end{claim}
\begin{proof}
First observe that the claim clearly holds for any $k$ and $d=0$, since for this case $\ub_\can^p(k,d) \geq k$ and by definition we have that $u_2(s_k(\dep,\dep)) \leq  k$. We now prove by induction on the number of rounds $k$ that the claim holds for $d\geq 1$. Note that the claim holds for the base case, $k=1$, since $u_2(s_1(x_1,x_2-d))\leq \max \{p,\can\} $ for every $d\geq 1$. We assume the claim holds for any $0<k'\leq k-1$ and prove it for $k$. If $u_2(s_k(\dep,\dep-d))=k\can$ we are done. Else, 
\begin{align*}
u_2(s_k(\dep,\dep-d)) &= p(1+u_2(s_{k-1}(\dep,\dep -d+1)))  
 + (1-p)(u_2(s_{k-1}(\dep+1,\dep-d)))
\end{align*}
By the assumption that $u_2(s_k(\dep,\dep-d))>k \can$ and using Claim \ref{clm:app:keeps_comp} we have that $u_2(s_{k-1}(\dep,\dep -d+1)) > (k-1) \can$. Thus, by the induction hypothesis (or our observation for $d=0$ in case $d$ was $1$), we have that $u_2(s_{k-1}(\dep,\dep -d+1)) \leq \ub_\can^p(k-1,d-1)$. By the induction hypothesis we also have that  $u_2(s_{k-1}(\dep+1,\dep -d)) \leq \max \{\ub_\can^p(k-1,d+1), (k-1) \can\}$. We distinguish between two cases depending on the two possible upper bounds on $u_2(s_{k-1}(\dep+1,\dep -d))$:
%\begin{itemize}
%\item 

If $u_2(s_{k-1}(\dep+1,\dep-d)) \leq \ub_\can^p(k-1,d+1) $:
\begin{align*}
u_2&(s_k(\dep,\dep-d))  \leq p\left(1+\ub_\can^p(k-1,d-1) \right) 
   + (1-p) \cdot  \ub_\can^p(k-1,d+1) \\
&= 2p + (k-2)\can 
   + (\frac{p}{1-p})^{d-1} \cdot \left(p+\frac{p^2}{1-p}\right) \cdot \big((k-2)(1-\can)+1\big) \\
&= 2p+ (\frac{p}{1-p})^{d} \big((k-2)(1-\can)+1\big)+ (k-2)\can\\
&< p+ (\frac{p}{1-p})^{d} \big((k-1)(1-\can)+1\big)+ (k-1)\can = \ub_\can^p(k,d)
\end{align*}
Else, $u_2(s_{k-1}(\dep+1,\dep-d)) = (k-1)\can$:
\begin{align*}
u_2&(s_k(\dep,\dep-d))  \leq  p\left(1+\ub_\can^p(k-1,d-1)  \right) + (1-p)(k-1)\can \\
& <  p + p^2 - p \can  + (k-1)\can 
   + p(\frac{p}{1-p})^{d-1} \big((k-1)(1-\can)+1\big) \\
%& <  p +\frac{p^d}{(1-p)^{d-1}} \big((k-1)(1-\can)+1\big) + (k-1)\can \\
& < p+(\frac{p}{1-p})^d \big((k-1)(1-\can)+1\big) + (k-1)\can = \ub_\can^p(k,d)
\end{align*}
%\end{itemize}
\end{proof}

We can now use the previous claim to identify 
the magnitude of the
difference between the players' reputations for which they stop competing. We state the claim for player $2$ but a similar claim also holds for player $1$.
\begin{lemma} \label{lem:fixedp_drops}
If $d>\dfrac{\log(k)-\log(p-\can)}{\log(\frac{1-p}{p})} = d_\can^p(k)$ then the lower player in the games $G_k(\dep,\dep-d)$, $G_k(\dep-d,\dep)$ does not compete.
\end{lemma}
\begin{proof}
Claim \ref{clm:frac_utility} reduces the problem of finding when does the lower player quits to solving for $d$ such that $\ub_\can^p(k,d)< k \can$. This would be enough to conclude that player $2$ stops competing as we know by Claim \ref{clm:app:kq-non-compete} that once $u_2(s_k(\dep_1,\dep_1)) = k\can$ player $2$ does not compete at all. After some rearranging of $\ub_\can^p(k,d)< k \can$ we have that:
\begin{align*}
\big( \frac{p}{1-p})^d \big((k-1)(1-\can)+1 \big)< \can - p
\end{align*}
We now take logarithms and get that:
\begin{align*}
d \cdot \log( \frac{p}{1-p}) + \log((k-1)(1-\can)+1) < \log(\can - p)
\end{align*}
Therefore,
%\begin{align*}
$d>\dfrac{\log((k-1)(1-\can)+1))-\log(p-\can)}{\log(\frac{1-p}{p})}$
%\end{align*}
and the claim follows.
\end{proof}

Next, we compute for how long the players are
expected to compete until the absolute value of the difference between their 
reputations becomes greater than the previously computed bound.
To do this, we study this difference as it performs a biased
random with a reflecting barrier at $0$:

\begin{lemma} \label{lem:fixedp_exp}
In the game $G_k(\dep_1,\dep_2)$, the expected number of rounds the players 
compete until the absolute value of the difference between 
their reputations is at least $d$ is at most $\frac{d}{1-2p}$.
\end{lemma}
\begin{proof}
% \soedit
We consider a particle undergoing
a biased random walk in which the probability of moving to the
left is $p$ and the probability of moving to the right is $1-p$, as before.
Since this particle tracks the absolute value of the
difference between the players' reputations, the walk
we are studying has a reflecting barrier on $0$.
This implies that when the particle reaches $0$, in the next step
it always goes to $1$.

Our analysis will thus be based on studying the expected time it takes
for the particle to reach the value $d$, starting from a value below $d$.
Clearly this expected time is maximized when the particle starts at
$0$, corresponding to an initial reputation difference of $0$.
Thus, we 
compute a bound on the expected number of rounds it takes 
players with identical
reputations to reach a difference of $d$ in their reputations.

The expected time it takes the particle to reach $d$ starting at $0$
when there is
a reflective barrier is upper bounded by the expected time it takes it
to reach $d$ starting at $0$ when there is no such barrier.
To see why, we invoke a standard argument in which we
imagine both walks being governed by the random flips of a
coin with bias $p$, and
we compare between the trajectory of the particle in
these two walks for the same random sequence of coin-flip outcomes.
We note that if the particle
reaches $d$ in the walk without the barrier it has to be the case that
it also reached $d$ using a prefix of the same sequence of coin-flip outcomes
in the walk with the barrier.

Finally, we use the fact
that the expected number of rounds required for a particle
performing this walk to reach $d$ starting from $0$, without a barrier at $0$,
is $\dfrac{d}{1-2p}$ (\cite{feller-vol1}).
To see why this is the case, let $E_i$ be the expected time for the walk to reach $i$. If $E_1$ is well defined, then $E_i=i \cdot E_1$. Also, $E_1=1+(1-p)\cdot 0 + p\cdot E_2= 1+2p\cdot E_1$. Therefore, $E_1=\dfrac{1}{1-2p}$ and $E_d=\dfrac{d}{1-2p}$.
\end{proof}

We are now ready to prove the following theorem.

\begin{theorem} \label{thm:fixedp_poa}
For fixed $p<\can$ the performance ratio of the game $G_k(\dep_1,\dep_2)$ is at least $1 - \theta(\dfrac{\log(k)}{k})$.
\end{theorem}
\begin{proof}
Let $R$ be a random variable equals to the number of rounds 
for which the players compete in the game $G_k(\dep_1,\dep_2)$. 
We can use it to compute the social welfare as follows:
\begin{align*}
u(s_k(\dep_1,\dep_2)) &= \sum_{r=1}^k Pr(R=r)(k+(k-r)\can) \btc
\atc = k(1+\can)-\can \sum_r Pr(R=r)\cdot r
\end{align*}
Now, to compute a lower bound on the social welfare, we should compute an upper bound on $\sum_{r=1}^k Pr(R=r)\cdot r$. We claim that $\sum_{r=1}^k Pr(R=r)\cdot r < \dfrac{d_\can^p(k)}{1-2p}$. The reason is that either the lower player quits when the difference between the players' reputation is $d_\can^p(k)$, as we proved in Claim \ref{clm:frac_utility}, or the lower player might decide to quit earlier in the game. 
In any case, the expected number of rounds the players compete until the lower player drops is at most $\dfrac{d_\can^p(k)}{1-2p}$, by Lemma \ref{lem:fixedp_exp}. Therefore, the performance ratio is at least 
$\dfrac{k(1+\can) - \dfrac{d_\can^p(k)}{1-2p} \can}{k(1+\can)}  = 1 - \theta(\dfrac{\log(k)}{k})$.
\end{proof}

\end{appendix}

\end{document}